\documentclass{article}

\usepackage{arxiv}
\usepackage[utf8]{inputenc} 
\usepackage[T1]{fontenc}    
\usepackage{hyperref}       
\usepackage{url}            
\usepackage{booktabs}       
\usepackage{amsfonts,amsmath,amsthm}       
\usepackage{nicefrac}       
\usepackage{microtype}      
\usepackage{lipsum}		
\usepackage{graphicx}
\usepackage{natbib}
\usepackage{doi}
\usepackage[skins]{tcolorbox}
\usepackage{tikz}
\usepackage{tkz-euclide}
\usepackage{pgfplots}
\pgfplotsset{compat=1.17}
\usepgfplotslibrary{fillbetween}
\usetikzlibrary{shapes.geometric, arrows.meta, positioning, fadings, decorations.pathmorphing, backgrounds, shadows,patterns,calc,3d,bending}
\usepackage{adjustbox,subcaption}
\usepackage{amssymb}

\newtheorem{definition}{Definition}

\newtheorem{theorem}{Theorem}

\newtheorem{proposition}{Proposition}

\newtheorem{remark}{Remark}

\title{Metric-Topology Factorization: A Computational Framework for Hippocampal-Neocortical Intelligence}

\author{ \href{https://orcid.org/0000-0003-2067-2763}{\includegraphics[scale=0.06]{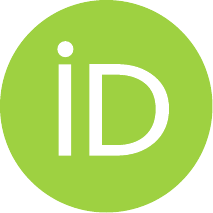}\hspace{1mm}Xin Li}\thanks{This work was partially supported by NSF IIS-2401748 and BCS-2401398. The author has utilized several GenAI models to aid in the development of theoretical ideas and the generation of visual illustrations presented in this paper.} \\ Department of Computer Science\\ University at Albany\\ Albany, NY 12222 \\ \texttt{xli48@albany.edu} }

\date{}

\hypersetup{
pdftitle={A template for the arxiv style},
pdfsubject={q-bio.NC, q-bio.QM},
pdfauthor={David S.~Hippocampus, Elias D.~Striatum},
pdfkeywords={First keyword, Second keyword, More},
}

\begin{document}
\maketitle

\begin{abstract}
How does the brain achieve stability and plasticity in a world whose semantic structure is topologically complex and continually shifting?
We propose that intelligence arises from a computational separation between topology and metric structure, formalized as \emph{Metric-Topology Factorization (MTF)}.
We argue that semantically complex environments cannot, in general, be represented by a single globally contractive geometry.
This \emph{geometric incompleteness} implies that any fixed smooth representational metric inevitably encounters saddle-type obstructions under topological shifts.
Intelligence therefore requires a factorization between discrete \emph{topological indexing}, which selects contextual regime, and continuous \emph{metric condensation}, which induces locally contractive inference within that regime.
In the brain, this separation corresponds to a hippocampal–neocortical division of labor embedded in a hierarchical cortical architecture.
The hippocampus generates sparse contextual signatures that index global manifold identity, while the neocortex implements sequential geometric untangling.
Along the ventral stream, a dynamic-programming(DP)-like hierarchy progressively quotients nuisance symmetries (translation, scale, rotation), converting a non-convex sensory ``maze'' into a linearly separable representational ``bowl.''
Replay and consolidation amortize these transport-like transformations offline, enabling rapid geometric switching rather than destructive re-optimization during task shifts.
We further interpret dreaming as the topological sampling phase of memory-amortized inference.
During REM sleep, stochastic hippocampal traversal of the cognitive graph exposes latent manifold structure, allowing cortical metrics to be regularized and recalibrated.
Consciousness corresponds to the online resolution of topological uncertainty into locally stable geometric inference: states enter awareness when their embedding is not yet fully amortized.
Finally, we view five major evolutionary transitions- namely sensorimotor control, allocentric navigation, episodic memory, social cognition, and language, as successive expansions of topological complexity requiring increasingly sophisticated separation between indexing and metric shaping.
Across wakefulness, sleep, and evolution, intelligence emerges not as deeper search, but as the controlled deployment and continual recalibration of context-specific geometries that transform global navigation problems into locally stable dynamics.
\end{abstract}

\section{Introduction}

How does the brain remain stable enough to preserve lifelong memories
while remaining plastic enough to learn new environments, social structures,
and languages \cite{buzsaki_rhythms_2006}?
The fundamental stability-plasticity dilemma has been central to computational neuroscience \cite{trappenberg2009fundamentals},
appearing in studies of hippocampal remapping \cite{buzsaki_hippocampo-neocortical_1996},
cortical consolidation \cite{squire_memory_2015},
catastrophic interference \cite{mccloskey_catastrophic_1989},
and systems-level replay.
Nevertheless, a striking asymmetry persists:
the dilemma appears substantially more severe in \emph{artificial neural networks}
than in biological systems \cite{kirkpatrick_overcoming_2017}.
Why do weight-based artificial systems readily overwrite prior knowledge,
while biological brains routinely incorporate new contexts
without erasing old ones?
We propose that the answer lies in a geometric constraint.
When the semantic state space of an organism becomes topologically complex,
no single smooth representational geometry can support globally
funnel-shaped inference.
Purely metric learning, reshaping a single energy landscape,
cannot eliminate separatrices and basin boundaries that arise from
nontrivial topology \cite{kulis2013metric}.
We formalize this limitation as a form of \emph{geometric incompleteness}:
for manifolds with nonvanishing intermediate homology,
any smooth energy function must contain saddle-type obstructions \cite{milnor1963morse}.
Under such conditions, failures of learning reflect not insufficient optimization,
but a mismatch between topological structure and metric representation.


Biological intelligence appears to resolve this constraint
through architectural factorization \cite{mcclelland_why_1995}.
The hippocampus performs rapid, sparse context discrimination and remapping,
while the neocortex supports gradual abstraction and consolidation
\cite{buzsaki_memory_2013}.
Rather than viewing this division as a mere speed–capacity tradeoff,
we interpret it as a separation between
\emph{topological indexing} and \emph{metric condensation}.
The hippocampus generates low-dimensional signatures that identify
the global context or orientation of experience \cite{teyler1986hippocampal}.
The neocortex, conditioned on this index,
shapes a locally contractive geometry that supports efficient inference
within the selected context \cite{mountcastle_columnar_1997}.
Replay and systems consolidation refine these context-indexed geometries,
implementing a form of memory-amortized inference (MAI)
that gradually converts high-entropy exploration
into low-entropy geometric stability.
Under the MAI framework \cite{li_beyond_2025},
contextual remapping and parity-like alternations
are handled not by overwriting a single global representation,
but by switching between discrete geometric regimes.
Catastrophic interference \cite{mccloskey_catastrophic_1989} becomes a symptom
of missing topological factorization
rather than an inevitable consequence of distributed representation.
Intelligence does not scale by deepening search,
but by expanding the capacity to identify \emph{which} structure one is in
before optimizing \emph{within} it \cite{li2026beyond}.

We extend this perspective to dreaming \cite{hobson1975sleep,hobson1977brain,revonsuo2000reinterpretation} and consciousness \cite{tononi_integrated_2016,tononi2004information,lycan1996consciousness}.
If wakefulness emphasizes metric stability within indexed contexts,
rapid eye movement (REM) sleep appears to temporarily relax metric constraints,
allowing stochastic traversal of the hippocampal cognitive graph \cite{muller_hippocampus_1996}.
We propose that dreams implement topological sampling:
a random-walk-like exploration decoupled from immediate sensory input and motor output.
The random exploration allows the neocortex to regularize and recalibrate
its metric tensor, preventing overfitting to recent experience \cite{buzsaki_hippocampo-neocortical_1996}.
In this sense, dreaming is not epiphenomenal,
but a necessary maintenance phase for geometric factorization \cite{hinton__1995}.
Consciousness, in turn, corresponds to the real-time execution
of metric amortization \cite{mashour_conscious_2020}:
the active reduction of topological uncertainty
into locally contractive inference centered on a stable self-referential basis.
What enters awareness are precisely those states
whose geometry is still being resolved.
Fully amortized processes recede into automaticity;
unresolved topological tensions become conscious content \cite{baars1993cognitive}.

Finally, we argue that major evolutionary advances in intelligence
can be interpreted as successive expansions of this principle \cite{bennett_brief_2023}.
Early sensorimotor systems implement local metric contraction
through closed-loop control.
Allocentric navigation introduces explicit topological indexing
via hippocampal cognitive maps \cite{tolman_cognitive_1948}.
Episodic memory and replay enable amortized refinement of
context-specific geometries \cite{tulving_episodic_2002}.
Social cognition scales indexing to multi-agent relational state spaces \cite{frith_social_2007},
and language introduces portable, compositional keys that permit
cross-individual transfer of condensed structure \cite{chomsky_three_1956}.
Each stage corresponds to increasing topological complexity
and a corresponding enhancement of indexing and geometric control mechanisms.
Across these transitions,
the same structure-before-specificity logic recurs:
identify the manifold first, then shape its metric.
Therefore, the stability–plasticity dilemma,
the evolutionary scaling of intelligence,
and the phenomenology of awareness
can be viewed as different expressions
of the same geometric principle.

\paragraph{Summary of Contributions.}
This work makes five main contributions.

\begin{enumerate}

    \item \textbf{Geometric Reformulation of the Stability–Plasticity Dilemma.}
    We formalize \emph{geometric incompleteness},
    showing that semantically complex manifolds generically induce
    saddle-type obstructions in any structurally stable smooth energy landscape.
    This reframes catastrophic interference
    as a structural consequence of topology.

    \item \textbf{Metric–Topology Factorization (MTF).}
    We propose a computational principle in which
    discrete topological indexing (context identification)
    is separated from continuous metric condensation
    (local geometric shaping),
    enabling reusable inference across incompatible regimes.

    \item \textbf{Biological Mapping to Hippocampal–Neocortical Architecture.}
    We interpret hippocampal remapping
    as topological indexing
    and cortical consolidation as metric condensation.
    Replay becomes memory-amortized refinement
    of context-indexed geometries.

    \item \textbf{A Geometric Account of Dreaming and Consciousness.}
    We propose that REM sleep implements stochastic
    topological sampling that regularizes metric structure,
    while consciousness reflects ongoing metric amortization
    of unresolved topological uncertainty.

     \item \textbf{Evolutionary Scaling of Intelligence.}
    We show that five major evolutionary breakthroughs,
    sensorimotor control, allocentric navigation,
    episodic memory, social cognition, and language,
    can be understood as expansions of indexing capacity
    and geometric control over increasingly complex state spaces.

\end{enumerate}


\section{Background and Motivation}

\subsection{Thermodynamic and Evolutionary Constraints on Complementary Learning Systems}
\label{sec:thermo_evo}

Why would biological intelligence adopt a separation between two complementary learnings systems \cite{mcclelland_why_1995}: hippocampus for topological indexing and neocortex for metric condensation?
We argue that this architectural pattern
is not arbitrary,
but reflects fundamental thermodynamic
and evolutionary constraints \cite{gould1982exaptation}.

\paragraph{Energy minimization and computational cost.}

Neural computation is metabolically expensive \cite{dayan_theoretical_2001}.
Action potentials, synaptic transmission,
and plasticity consume substantial ATP.
Brute-force search over high-dimensional,
non-convex state spaces
requires sustained active computation \cite{friston_action_2010}:
maintaining competing hypotheses,
simulating alternative trajectories,
and resolving uncertainty in real time.
Such operations scale poorly with dimensionality
and rapidly increase energetic demand \cite{bellman_dynamic_1966}.
A more efficient strategy is to shift part of this computational burden
offline, converting experience into structure that is available \emph{before}
the next problem arrives. Borrowing the term from complexity theory, we call
such pre-positioned, instance-independent structure \emph{advice}
\cite{arora2009computational}. Advice is not a repackaging of time and space:
it is provably non-fungible with them. Nonuniform advice decides languages
that no time or space budget whatsoever can decide -- the class
$\mathrm{P/poly}$ contains undecidable languages, since a unary encoding of
the halting problem is settled by a single advice bit per input length
\cite{arora2009computational}. Consolidation and replay are therefore not a
space-for-time exchange but the \emph{manufacture of advice}: they convert
spent time and space into structure that later inference consults for free
\cite{gershman_amortized_2014}.
Rather than repeatedly solving a complex search problem at inference time,
the system can gradually embed invariant structure
into its synaptic organization \cite{hebb2005organization}.
Over experience, regularities of the environment
become encoded in the geometry of neural connectivity.
Subsequent inference can then proceed through fast,
approximately gradient-like dynamics
within a pre-structured representation space \cite{geman_neural_1992}.
In this way, the energetic cost of resolving structure
is paid during learning and consolidation,
reducing the metabolic burden of future decisions.

\paragraph{Time pressure and survival.}

Evolution imposes a second constraint \cite{bennett_brief_2023,daw2005uncertainty}:
decisions must often be made under severe temporal limits.
Predator avoidance, motor coordination,
and social interaction require responses
on timescales that preclude extended deliberation or simulation.
Sequential evaluation of alternative trajectories
is too slow when survival depends on rapid action.
Under such conditions,
organisms benefit from pre-structured internal models
that allow behavior to converge quickly toward appropriate responses.
Rather than computing solutions from scratch in real time,
the nervous system can embed environmental regularities
into long-term synaptic organization
through experience-dependent plasticity
and systems consolidation \cite{frank2006sleep,tononi2003sleep}.
Replay during sleep and quiet wakefulness
appears to refine these internal models offline,
when sensory demands are reduced \cite{wilson1994reactivation}.
As a result, online behavior can operate in a low-latency regime,
where actions follow rapidly from structured neural dynamics
rather than from active search.
Evolution thus favors mechanisms that exchange real-time computation
for structured representation:
what would otherwise require slow sequential exploration
is replaced by fast convergence within a learned internal geometry.

\paragraph{Continuity of the physical world and Exaptation}

The physical environment is organized as a continuous dynamical system.
Spatial relations, trajectories, and transformations
are primary features of the world,
whereas discrete object boundaries
are emergent regularities inferred by observers.
Early nervous systems evolved
to solve spatial navigation problems \cite{tolman_cognitive_1948,okeefe_hippocampus_1978}:
locating resources,
avoiding threats,
and coordinating movement.
These functions require representing
continuous spatial manifolds
and computing trajectories within them \cite{muller_hippocampus_1996}.
As perceptual and social complexity increased,
evolution appears to have reused
these spatial computation mechanisms
for higher-dimensional domains \cite{behrens_what_2018}.
Object identity, social roles,
and abstract concepts
can be represented as coordinates
within learned metric spaces \cite{epstein2017cognitive}.
Categorization becomes a form of navigation
to stable regions of semantic space.
Rather than inventing a new computational principle,
evolution extended existing hippocampal–cortical loops
to increasingly abstract domains.
This process is consistent with exaptation \cite{gould1982exaptation},
in which biological systems repurpose established mechanisms
for new functions without redesigning them from first principles.
Circuits that originally evolved to support spatial navigation,
path integration, and relational mapping \cite{okeefe_hippocampus_1978}
could be reutilized to organize non-spatial information.
As sensory systems expanded and social and symbolic environments grew more complex,
the same relational machinery that once tracked locations in physical space
could encode positions in abstract spaces of objects, tasks, and meanings \cite{behrens_what_2018}.
Empirical evidence suggests that neural systems supporting navigation
also represent continuous abstract dimensions,
such as conceptual or relational variables,
using spatial-like coding principles \cite{constantinescu2016organizing}.

\paragraph{Intelligence and reuse of spatial machinery.}

From an evolutionary perspective,
the computational machinery for spatial navigation
constitutes an ancestral substrate of cognition
\cite{okeefe_hippocampus_1978,buzsaki_rhythms_2006}.
The hippocampal–entorhinal system constructs relational maps
that support flexible navigation and inference
\cite{moser_place_2008,behrens_what_2018}.
When mammals later evolved increasingly complex sensory systems,
including high-resolution vision,
evolution did not need to invent a fundamentally new
``categorization engine.''
Instead, high-dimensional sensory signals
could be routed into pre-existing relational circuits.
Object identity could then be represented
as a stable position within an abstract relational space.
Growing evidence indicates that neural systems originally involved in spatial navigation
also encode non-spatial and conceptual dimensions
using spatial-like representational principles
\cite{constantinescu2016organizing,aronov2017mapping}.
In this view, continuous relational structure (``where'')
precedes discrete identity (``what'') \cite{ungerleider_whatand_1994}.
The brain first organizes experience
in terms of structured spatial and relational dynamics,
and invariant clusters gradually emerge within this space,
giving rise to stable object representations.
Discrete concepts are therefore not primitive symbolic units,
but fixed points within an underlying relational geometry.
Categorization can thus be understood as a reuse
of spatial navigation machinery:
the brain treats identities as locations it can reach.
This evolutionary reuse of spatial machinery suggests that the central computational challenge of intelligence can be reframed geometrically \cite{nakahara_geometry_2018}: how does the brain transform a tangled relational landscape into a structure that supports rapid, stable convergence?
\vspace{-0.1in}
\subsection{Geometric Incompleteness}
\label{sec:geom_incomplete_sub}
\label{sec:2}

\paragraph{The Bowl-Maze Analogy}
To build intuition for geometric incompleteness,
consider two contrasting views of learning (Fig. \ref{fig:teaser}).
In a \emph{maze view},
the environment is treated as a fixed combinatorial structure.
Learning consists of discovering a path through this structure \cite{minsky_steps_1961}.
The geometry of the space is assumed to be given;
the learner merely searches within it.
Improvements arise from better exploration strategies
or more efficient optimization \cite{russell1995modern}.
In a \emph{bowl view},
the learner reshapes the geometry itself.
Rather than searching for a path,
the system modifies its internal metric so that the goal
becomes the bottom of a locally contractive basin.
Inference then proceeds by descent within this shaped geometry \cite{geman_neural_1992}.
Many successful learning systems implicitly rely on this principle:
they aim to transform a rugged landscape into a smooth basin
through representation learning \cite{bengio_representation_2013}.
These two views coincide when the underlying state space
is topologically simple.
If the semantic manifold $\mathcal{M}$ is contractible,
a sufficiently expressive metric transformation
can in principle induce a globally funnel-shaped landscape \cite{amari1998natural}.
Search can then be replaced by geometry.
However, the equivalence breaks down when $\mathcal{M}$
is semantically complex.

Consider a parity-twisted or Möbius-like state space \cite{hirsch1976differential},
in which a traversal around a nontrivial cycle
reverses orientation.
Any attempt to impose a single global ``bowl''
must introduce a seam or separatrix that partitions
distinct basins of attraction.
No smooth reshaping of the metric can eliminate
the global identification constraint \cite{milnor1963morse}.
The learner may smooth local curvature,
but the topology forces unstable boundaries.
The above observation with the Maze-Bowl analogy reframes the stability-plasticity dilemma \cite{wang_comprehensive_2024}.
When catatrophic interference occurs after contextual change,
the failure may not reflect insufficient capacity
or imperfect optimization \cite{mccloskey_catastrophic_1989}.
Instead, it may signal that the system is attempting
to compress incompatible topological regions
into a single metric basin.
The Bowl-Maze analogy therefore motivates a separation of roles.
If topology imposes unavoidable basin boundaries,
then learning cannot rely solely on metric shaping \cite{kulis2013metric}.
It must also identify which topological regime
the current experience belongs to.
In the following subsection,
we formalize this intuition using Morse theory \cite{milnor1963morse}
and show that saddle-type obstructions
are generically unavoidable on semantically complex manifolds.

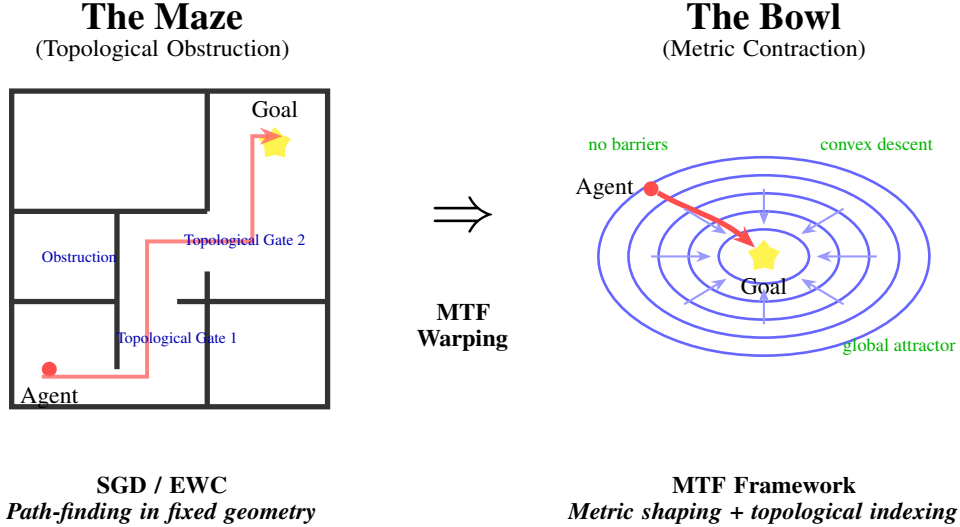
\begin{figure*}[h]
    \centering
   \begin{tikzpicture}[
    node distance=0.5cm,
    agent/.style={circle, fill=red!70, inner sep=2pt},
    goal/.style={star, star points=5, fill=yellow!80, inner sep=3pt},
    wall/.style={line width=2pt, black!80},
]

\begin{scope}[local bounding box=maze]
    \node[font=\Large\bfseries, anchor=north] at (0, 3.5) {The Maze};
    \node[font=\small, text width=5.5cm, align=center, anchor=north] at (0, 3.) {
        (Topological Obstruction)
    };
        
    \draw[wall] (-2, -2) rectangle (2.2, 2.2);
        
    \draw[wall] (-0.6, -1.5) -- (-0.6, 0.6);      
    \draw[wall] (0.6, -2) -- (0.6, -0.2);       
    \draw[wall] (0.6, 0.6) -- (0.6, 2.2);         

    \draw[wall] (-2, -0.6) -- (-0.6, -0.6);     
    \draw[wall] (0.2, -0.6) -- (2.2, -0.6);       
    \draw[wall] (-2, 0.6) -- (0.6, 0.6);        
    
    \node[agent, label={[font=\small]below :Agent}] (agent_maze) at (-1.5, -1.5) {};
    \node[goal, label={[font=\small]above :Goal}] (goal_maze) at (1.5, 1.5) {};

    
    \draw[red!60, line width=1.5pt, -{Stealth[length=3mm]}, opacity=0.8] 
        (-1.6, -1.6) -- 
        (-0.2, -1.6) -- 
        (-0.2, 0.2)  -- 
        (1.2, 0.2)   -- 
        (1.2, 1.6)   -- 
        (1.6, 1.6);

    \node[font=\tiny, text=blue!70!black] at (-1.1, 0.0) {Obstruction};
    \node[font=\tiny, text=blue!70!black] at (0.2, -1.1) {Topological Gate 1};
    \node[font=\tiny, text=blue!70!black] at (1.1, 0.2) {Topological Gate 2};

    \node[font=\small\bfseries, text width=5cm, align=center, anchor=north] at (0, -2.8) {
        SGD / EWC\\
        \textit{Path-finding in fixed geometry}
    };
\end{scope}

\node[font=\Huge, anchor=center] at (4., 0.5) {$\Rightarrow$};
\node[font=\small\bfseries, text width=2cm, align=center, anchor=north] at (4., -0.5) {
    MTF\\Warping
};

\begin{scope}[shift={(8,0)}, local bounding box=bowl]
    \node[font=\Large\bfseries, anchor=north] at (0, 3.5) {The Bowl};
    \node[font=\small, text width=5.5cm, align=center, anchor=north] at (0, 3.) {
        (Metric Contraction)
    };
    
    \foreach \r/\op in {2.2/0.05, 1.8/0.1, 1.4/0.15, 1.0/0.25, 0.6/0.35} {
        \draw[blue!60, fill=blue!\op!white, line width=1pt] (0, 0) ellipse ({\r} and {\r*0.6});
    }
    
    \node[agent, label={[font=\small]left:Agent}] (agent_bowl) at (-1.5, 0.9) {};
    
    \node[goal, label={[font=\small]below:Goal}] (goal_bowl) at (0, 0) {};
    
    \draw[red!70, line width=2pt, -{Stealth[length=3mm]}] 
        (agent_bowl) to[out=-30, in=135] (goal_bowl);
    
    \foreach \angle in {0, 45, 90, 135, 180, 225, 270, 315} {
        \draw[-{Stealth[length=2mm]}, blue!40, line width=0.8pt] 
            ({1.5*cos(\angle)}, {0.9*sin(\angle)}) -- 
            ({0.7*cos(\angle)}, {0.42*sin(\angle)});
    }
    
    \node[font=\scriptsize, text=green!70!black] at (1.8, -1.2) {global attractor};
    \node[font=\scriptsize, text=green!70!black] at (-1.8, 1.5) {no barriers};
    \node[font=\scriptsize, text=green!70!black] at (1.5, 1.5) {convex descent};
    
    \node[font=\small\bfseries, text width=8cm, align=center, anchor=north] at (0, -2.8) {
        MTF Framework\\
        \textit{Metric shaping + topological indexing}
    };
\end{scope}



\end{tikzpicture}
    \caption{\textbf{The Bowl-Maze Analogy: From path-finding to metric contraction.} 
    \emph{Left: The Maze (Topological Obstruction).} Conventional learning systems, such as stochastic gradient descent (SGD) and elastic weight consolidation (EWC), treat intelligence as navigation within a fixed, complex geometry. In this ``Maze''-based regime, topological features, such as holes, walls, or parity flips, manifest as local minima and dead-ends. Learning is a slow, procedural search for a valid path, which is easily invalidated by any change in the environment's topology. 
    \emph{Right: The Bowl (MTF).} Under the MTF framework, intelligence is redefined as the ability to shape the metric structure of the space. Instead of searching for paths, the agent performs \textit{topological indexing} to switch into a coordinate system where the solution emerges as a stable, global attractor. By warping the maze into a ``Bowl'', the agent replaces complex navigation with a simple downhill descent, ensuring that once the topology is identified, the solution is both reachable and consistent.}
    \label{fig:teaser}
\end{figure*}

\noindent\textbf{Topological constraints on smooth inference}
Learning systems often assume that improved optimization
can eliminate interference between tasks.
However, when the semantic state space of an organism
has nontrivial topology, this assumption fails \cite{mcclelland_why_1995} because
topological obstruction is structural rather than algorithmic.
Let $\mathcal{M}$ denote a smooth, compact manifold representing
the latent semantic state space of an environment.
We say that $\mathcal{M}$ is \emph{semantically complex}
if it possesses nonvanishing intermediate homology \cite{hatcher_algebraic_2002},
i.e.,
$\beta_k(\mathcal{M}) = \mathrm{rank}\, H_k(\mathcal{M};\mathbb{R}) > 0
\quad \text{for some } 1 \le k \le d-1$.
Intuitively, such manifolds contain holes, loops,
or separable regions that prevent global contraction
(e.g., a torus rather than a sphere) \cite{milnor1963morse}.
A learning system that relies on a single smooth energy landscape
$E:\mathcal{M}\to\mathbb{R}$
attempts to shape the geometry of $\mathcal{M}$
so that gradient flow
$\dot z = -\nabla_g E(z)$
produces \emph{saddle-free} descent toward a goal.
The following result formalizes why this is impossible
on semantically complex manifolds.
We state it with its hypotheses made explicit, because both matter:
the state must be \emph{intrinsically confined} to $\mathcal{M}$,
and the conclusion concerns the \emph{existence of critical points},
not the failure of convergence.

\begin{theorem}[Metric-Topological Incompleteness]
\label{thm:geom_incomplete}
Let $\mathcal{M}$ be a compact, smooth $d$-dimensional manifold
with nonvanishing intermediate homology,
i.e., $H_k(\mathcal{M}) \neq 0$ for some $1 \le k \le d-1$,
and suppose the system state is constrained to evolve on $\mathcal{M}$,
i.e., the flow $\dot z = -\nabla_g E(z)$ is the gradient flow of
$E$ with respect to a Riemannian metric $g$ on $\mathcal{M}$
(equivalently, $\mathcal{M}$ is an invariant manifold of the dynamics).
Then any Morse function $E:\mathcal{M}\to\mathbb{R}$
must possess at least $\beta_k(\mathcal{M})$ critical points
of index $k$.
In particular, no such $E$ admits a saddle-free descent:
any structurally stable smooth energy landscape
on $\mathcal{M}$ necessarily contains index-$k$ saddle-type equilibria,
whose stable and unstable manifolds furnish a local scaffold--flow splitting.
\end{theorem}

\begin{remark}[What the theorem does and does not assert]
\label{rem:scope_thm1}
Two qualifications are essential, and we state them explicitly because the
theorem is easily over-read.
\emph{(i) Forced saddles obstruct monotone descent, not convergence.}
Theorem~\ref{thm:geom_incomplete} guarantees that critical points of
intermediate index exist; it does \emph{not} guarantee that gradient flow
fails to reach a global minimum. The height function on the torus
$T^2$ ($\beta_1=2$) has one minimum, two saddles and one maximum, yet
gradient flow converges to the unique minimum from almost every initial
condition, because the stable manifolds of the saddles have measure zero.
Moreover, strict (nondegenerate) saddles are generically escapable:
perturbed gradient descent leaves them in polynomial time
\cite{ge2015escaping,jin2017escape}. The obstruction is therefore to
\emph{saddle-free, strictly monotone} descent, and to the existence of a
\emph{global} certificate of progress -- not to eventual convergence.
This distinguishes geometric incompleteness from logical undecidability,
where no escape route exists at any cost; the analogy in
Remark~\ref{rem:godel_analogy} is structural, not an equivalence of hardness.
\emph{(ii) The homology must be that of the space the state actually
inhabits.} If the dynamics run in an ambient representation space
$\mathbb{R}^D \supset \mathcal{M}$, that ambient space is contractible, all
its Betti numbers vanish, and no saddle is forced: the flow may simply
descend around the hole. The theorem therefore bites precisely when the
neural state is intrinsically confined to a topologically nontrivial
manifold -- as for head-direction cells on a ring ($S^1$), grid cells on a
torus ($T^2$), or task spaces with a parity/M\"obius identification -- and
not for a generic unconstrained network. Every persistence claim below is
relative to a specified state space, and we name it wherever we use it.
\end{remark}

\begin{remark}[Geometric Analogy to Gödel's Incompleteness]
\label{rem:godel_analogy}
Theorem~\ref{thm:geom_incomplete} is a structural (not an equivalent) counterpart to Gödel's First Incompleteness Theorem. In mathematical logic, Gödel demonstrated that any formal system possessing sufficient structural complexity (i.e., capable of expressing Peano arithmetic) must inevitably contain undecidable propositions, statements that are true but cannot be proven via local syntactic derivations within that system \cite{godel_formally_1962}. 
Similarly, under a geometric framework, Theorem~\ref{thm:geom_incomplete} dictates that any representational manifold $\mathcal{M}$ possessing sufficient topological complexity (i.e., non-trivial intermediate homology $H_k(\mathcal{M}) \neq 0$) must inevitably contain ``undecidable'' optimization states. These are the index-$k$ saddle points \cite{milnor1963morse}: critical equilibria where local, continuous gradient information (the metric derivative) is insufficient to resolve the global optimal path. 
Just as a formal logical system cannot certify its own consistency without appealing to a higher-order meta-language, an agent restricted to local metric descent possesses no \emph{intrinsic} certificate that its current basin is the correct one: at an index-$k$ saddle the gradient is silent about which side of the separatrix leads to the goal, and that one bit must be supplied from outside the flow.
The parallel should not be overstated. As noted in Remark~\ref{rem:scope_thm1}, a strict saddle is generically \emph{escapable} under noise, whereas a G\"odel sentence is not decidable at any cost; the two are conjugate in \emph{form} -- a missing certificate from outside the resolving mode -- but not in \emph{hardness}.
In both domains the resolution requires stepping outside the mode that stalls \cite{nakahara_geometry_2018}: logic appeals to a stronger meta-system, whereas biological intelligence supplies a discrete contextual index that selects the regime before metric descent proceeds within it.
\end{remark}

\paragraph{Sketch of intuition.}
Morse theory establishes a direct relationship between
the topology of a manifold and the critical points of smooth functions defined on it.
Specifically, the Morse inequalities relate
the number of critical points of index $k$
to the $k$-th Betti number $\beta_k = \dim H_k(\mathcal{M})$ \cite{milnor1963morse}.
If $\beta_k > 0$ for some intermediate index $k$,
then any Morse function on $\mathcal{M}$
must contain at least one critical point of that index.
Such critical points correspond to saddles under gradient flow.
Geometrically, intermediate homology implies the presence of
nontrivial cycles (``holes,'' ``handles,'' or ``twists'')
in the manifold.
These topological features cannot be removed by any smooth deformation \cite{edelsbrunner_computational_2010}.
As a result, no smooth energy function on $\mathcal{M}$ can be
\emph{saddle-free}: the descent must pass through critical regions whose
stable and unstable manifolds locally partition the state space into a
contracting scaffold and an expanding flow.
The key insight is that these saddle points are not accidental artifacts
of poor optimization or limited expressivity.
They are required by topology \cite{willard_general_2012}.
We emphasize the precise content of this claim (Remark~\ref{rem:scope_thm1}):
topology forces the \emph{existence} of index-$k$ critical points and hence
of separatrices, but it does not by itself force multiple basins of
attraction, nor does it prevent gradient flow from converging to a global
minimum from almost every initialization.
Nonconvexity is in this sense topologically inevitable \cite{chi2019nonconvex},
but nonconvexity is an obstruction to \emph{monotone, certificate-bearing}
descent rather than a proof that descent fails.

\noindent\textbf{Algorithmic consequences.}
Although the stable manifolds of saddle points are measure-zero
in continuous-time gradient flow,
their dynamical influence is disproportionately large \cite{danilova2022recent}.
In the vicinity of a saddle,
the Hessian of $E$ has both positive and negative eigenvalues,
producing directions of attraction and repulsion.
Under discretization, stochastic gradient noise,
or finite sampling,
trajectories exhibit critical slowing down,
sensitivity to initialization,
and unstable transitions between basins.
These effects create practical bottlenecks:
optimization may stall near flat saddle regions \cite{dauphin2014identifying},
small perturbations may induce abrupt basin switching,
and previously stable attractors may lose dominance
when new regions of the manifold are explored.
Such behavior reflects the unavoidable separatrices imposed by topology beyond the reach of algorimic stability.

\noindent\textbf{Geometric incompleteness.} We therefore refer to this
limitation as \emph{geometric incompleteness} \cite{li2026beyond}:
when the state is confined to a topologically nontrivial semantic manifold,
no fixed smooth representational geometry supports saddle-free, globally
certified descent. Topology induces unavoidable saddle structure, and
saddles induce dynamical bottlenecks -- critical slowing, sensitivity to
initialization, and separatrix-crossing under perturbation.
This is the \emph{intrinsic-topology} route to the necessity of indexing,
and it applies exactly where the neural state space genuinely carries
nontrivial homology (ring, torus, parity-twisted task spaces).
It is not, however, the only route, and it is not the one that explains
catastrophic interference in the general case. Interference arises when two
contexts demand \emph{mutually inconsistent maps on overlapping inputs} --
a conflict obstruction that requires no nontrivial homology whatever. We
develop this second, more general argument in
Section~\ref{sec:conflict}, and treat the homological argument of the
present section as the special case in which the conflict is enforced by the
intrinsic topology of the state space itself.
Both routes converge on the same conclusion:
intelligence must introduce an additional mechanism,
a discrete indexing or factorization into complementary learning systems \cite{mcclelland_why_1995},
that separates contextual identity from metric shaping.

\noindent\textbf{Biological interpretation.}
The hippocampus-neocortex division can be interpreted
as a biological response to geometric incompleteness \cite{mcclelland_why_1995}.
Rather than attempting to learn a single global landscape,
the brain separates: \emph{topological identification},
which determines the global context or orientation of experience,
from \emph{metric condensation},
which shapes local geometry within that context.
Under this view,
contextual remapping \cite{kubie2020hippocampal} does not represent a failure of stability,
but an adaptive response to unavoidable saddle structure.
By switching between discrete topological regimes,
the brain avoids attempting to force incompatible
regions of $\mathcal{M}$
into a single global basin.
In the following section,
we formalize this separation as
Metric-Topology Factorization (MTF)
and show how topological indexing \cite{teyler1986hippocampal}
enables context-specific geometric control.

\subsection{The Conflict Obstruction: Width and the Phase-Transition Floor}
\label{sec:conflict}

The argument of Section~\ref{sec:geom_incomplete_sub} is exact but narrow: it
requires the state to be confined to a manifold with nonvanishing intermediate
homology. Catastrophic interference, however, occurs routinely in
representation spaces that are contractible. We therefore give a second,
logically independent argument for the necessity of indexing -- one that
requires no topology at all, that directly explains interference, and that
yields a sharp quantitative prediction. We call it the \emph{conflict
obstruction}, and we regard it as the primary route; the homological argument
is then the special case in which the conflict is enforced by the intrinsic
topology of the state space.

\paragraph{Setup.} Let a task environment present $w$ contexts
$c \in \{1,\dots,w\}$, each equally weighted, each specifying a labeling rule
$r_c: \mathcal{X} \to \mathcal{Y}$ on a shared input space $\mathcal{X}$. We
say the contexts are \emph{maximally conflicting} if any two distinct rules
disagree everywhere: $r_c(x) \neq r_{c'}(x)$ for all $x \in \mathcal{X}$ and
all $c \neq c'$. We call $w$ the \emph{width} of the environment. A factorized
learner allocates $K$ cells; routing assigns each context to a cell, and each
cell applies a single context-independent map to the input.

\begin{definition}[Width]
\label{def:width}
The width $w$ of an environment is the number of mutually inconsistent
labeling regimes it contains, i.e., the minimum number of context-independent
maps required to realize all of its rules exactly.
\end{definition}

\begin{proposition}[Phase-transition error floor]
\label{prop:floor}
Under the above setup with hard (discrete) routing, the test error of any
$K$-cell learner satisfies
\[
\varepsilon(w,K) \;\ge\; \eta(w,K) \;:=\; \max\!\left(0,\ \frac{w-K}{w}\right),
\]
and the bound is attained by any balanced assignment of contexts to cells.
In particular the floor vanishes exactly when $K \ge w$, and the transition at
$K = w$ is sharp.
\end{proposition}

\begin{proof}
Any routing partitions the $w$ contexts into $K$ groups of sizes
$g_1,\dots,g_K$ with $\sum_j g_j = w$. A cell serving a group applies one map;
since the rules in the group pairwise disagree on every input, that map can
agree with at most one of them, and errs on the entire mass of the other
$g_j - 1$. Group $j$ carries mass $g_j/w$, contributing error
$(g_j-1)/w$. Summing, $\sum_j (g_j - 1)/w = (w - K)/w$ when $K \le w$,
independent of the partition. For $K \ge w$ each context receives its own
cell and the floor is zero.
\end{proof}

\paragraph{Interpretation.} Proposition~\ref{prop:floor} is the general form of
geometric incompleteness, and it is stronger than the homological argument in
three respects. It requires no assumption about the topology of the
representation space -- the floor appears in $\mathbb{R}^D$. It is
\emph{quantitative}: it predicts not merely that a single metric fails, but the
exact residual error it must incur. And it identifies the resource that must be
allocated -- $K$, the number of indexable regimes -- and the observable that
governs it -- $w$, the width. Catastrophic interference is then not a
consequence of holes in a manifold, but of \emph{under-allocated indexing
capacity}: $K < w$. Under-allocation, not distributed representation, is the
failure mode.

\paragraph{Empirical confirmation and a falsifiable prediction.} The floor
$\eta(w,K)$ is directly testable, and we have verified it in controlled
numerical experiments on planted-width data (contexts identified by a separable
marker; a cyclic-shift relabeling that renders any two contexts maximally
conflicting). Sweeping $K$ at fixed $w \in \{3,5,8\}$, the test error of a
hard-routed $K$-cell learner exhibits a high plateau for $K < w$, a sharp knee
located at $K = w$ for every $w$ tested, and a plateau at the oracle error for
$K \ge w$; a parameter-matched dense network shows no such $w$-locked knee,
confirming that the effect is a property of indexing capacity rather than of
raw expressivity. The measured floor respects $\eta(w,K)$ as a lower bound and
is tight near $K = w$. MTF thus makes a risky prediction that a capacity-based
account does not: \emph{the knee must track $w$, not the parameter count.}

\paragraph{The role of relaxation.} One further empirical observation
sharpens the theory. The floor of Proposition~\ref{prop:floor} is a property of
\emph{hard} routing. If routing is relaxed -- each context assigned a
distribution over cells, so that the effective map is a convex blend of the $K$
cell maps -- the learner can partially escape $\eta(w,K)$ whenever the $w$
rules are not linearly independent, because $K$ blended maps can span more than
$K$ distinct effective maps. Continuous relaxation is therefore not merely a
computational convenience for an intractable discrete assignment; it is a
genuine, if partial, escape from the discrete floor, available exactly to the
extent that the contexts share structure. This is the precise sense in which
the continuous and discrete halves of MTF are complementary rather than
redundant: indexing is indispensable when the regimes are genuinely
incompatible, and relaxation recovers capacity when they are not.

\section{Metric-Topology Factorization via Hippocampal-Neocortical Dialogue}
\label{sec:3}

\subsection{Metric-Topology Factorization}

Section~\ref{sec:2} showed that semantically complex state spaces
cannot support globally contractive inference
under a single smooth energy landscape.
If topology induces unavoidable saddle structure,
then stability cannot be achieved by metric shaping alone.
An additional mechanism is required.
We propose that biological intelligence resolves
this geometric limitation through
\emph{Metric-Topology Factorization (MTF)}:
a separation between continuous metric dynamics within a context,
and discrete indexing across contexts.

\noindent\textbf{Metric Collapse (Within-Context Dynamics).}
Given a fixed context,
neural activity evolves under gradient-like dynamics
$\dot z = -\nabla_g E(z)$,
driving states toward locally stable attractors \cite{amari1998natural}.
We refer to this contraction toward a basin
as \emph{metric collapse}.
It corresponds to active inference during wakefulness \cite{friston_active_2017},
where prediction errors are minimized
and neural trajectories settle into stable representations.
Importantly, metric collapse operates
within an existing geometric regime \cite{friston_action_2010}.
It reshapes activity distributions
but does not eliminate the global separatrices
imposed by topology.

\noindent\textbf{Topological Indexing (Across-Context Separation).}
To prevent interference between incompatible regimes,
the system must first determine
which semantic manifold is currently active.

\begin{definition}[Topological Index]
A topological index is a discrete signature
$\Sigma = \mathcal{S}(x)$
that partitions the latent state space
into context-dependent regions
$\{\mathcal{M}_\Sigma\}$.
\end{definition}

The index does not encode detailed geometry.
Rather, it identifies global orientation or regime
(e.g., environment identity, parity state, or social context).
In biological systems,
rapid hippocampal remapping provides such indexing \cite{teyler1986hippocampal}:
sparse, orthogonal firing patterns
signal which contextual manifold
the neocortex should operate within \cite{buzsaki_memory_2013}.

\noindent\textbf{Metric Condensation (Schema Formation).}
Repeated collapse within a context
leads to consolidation \cite{squire_memory_2015}.
Over time, detailed metric structure
within a basin is compressed into a stable schema.
We refer to this structural compression
as \emph{metric condensation}.
Unlike transient collapse,
condensation reflects a longer-term transformation \cite{tse_schemas_2007}:
representations become increasingly invariant
to within-context variability.
The system no longer stores fine-grained trajectories,
but retains only their stable geometric core.
Neocortical consolidation and replay
provide a plausible biological substrate
for this gradual compression \cite{stickgold_sleep-dependent_2005}.

\noindent\textbf{Factorized Inference.}
Inference under MTF proceeds in two stages:
1) \textit{Index:} determine the active context $\Sigma(x)$;
2) \textit{Collapse:} perform contractive inference within the selected regime.
Formally,
$\mathcal{F}(x)
=\mathcal{G}_{\Sigma(x)}(\Phi_\theta(x))$,
where $\Phi_\theta$ denotes shared feature extraction
and $\mathcal{G}_{\Sigma}$ denotes context-specific
collapse dynamics.
By separating context identity from metric shaping,
the system avoids forcing incompatible regimes
into a single global geometry \cite{mcclelland_why_1995}.
Stability is maintained through discrete indexing,
while plasticity is preserved through local adaptation \cite{abraham2005memory}.

\paragraph{Hierarchical inference as metric amortization.}

Under the MTF framework,
the process of hierarchical inference \cite{lee_hierarchical_2003} can be interpreted
as sequential metric amortization.
Each cortical stage reduces a subset of variability,
thereby decreasing effective search complexity.
Instead of solving a single high-dimensional
non-convex optimization,
the system performs a series of local transformations,
each simplifying the manifold geometry \cite{RiesenhuberPoggio1999Hierarchical}.
In the continuous limit,
this hierarchical composition can be viewed
as approximating a smooth flow
that transports sensory states
toward more linearly organized representations \cite{dicarlo_how_2012}.
Learning adjusts the parameters of each $f_\ell$
so that trajectories through the hierarchy
follow progressively straighter paths,
facilitating simple downstream readout.

\noindent\textbf{Biological Mapping}
The hippocampal-neocortical architecture \cite{buzsaki_hippocampo-neocortical_1996}
naturally implements this factorization (Fig. \ref{fig:horizontal_mtf}).
1) \textit{Hippocampus (Indexing):}
    Rapid remapping generates sparse signatures
    that partition experience into discrete contexts \cite{teyler1986hippocampal}.
    These signatures function as topological keys.
2) \textit{Neocortex (Condensation):}
    Within each indexed context,
    cortical networks gradually refine a contractive geometry
    through synaptic plasticity and replay-driven consolidation \cite{tse_schemas_2007}.
3) \textit{Replay (Amortization):}
    Recurrent reactivation refines
    the metric associated with each signature,
    implementing memory-amortized inference across episodes \cite{gershman_amortized_2014}.
Under this framework,
catastrophic interference \cite{mccloskey_catastrophic_1989}
arises when distinct relational regimes are compressed
into a single representational geometry.
When contextual separation is absent,
learning forces incompatible structures
into a shared basin of attraction,
producing interference rather than integration.

Conversely, remapping and context-dependent coding
should not be viewed as failures of stability,
but as principled responses to structural constraints \cite{yamins_using_2016}:
when the underlying relational manifold changes,
the representational geometry must change with it.
If stability requires separating relational structure
before shaping local geometry,
a natural question follows:
how does the brain algorithmically transform
high-dimensional sensory input
into structured, low-interference representations \cite{simoncelli2001natural}?
In the following section,
we examine how this transformation is implemented
through hierarchical processing in the ventral stream,
and how this sequential decomposition
converts a globally tangled landscape
into locally contractive representations
that support efficient inference.

\subsection{Hierarchical Untangling in the Ventral Stream}
\label{sec:ventral_hierarchy}

\paragraph{The curse of dimensionality.}

Directly mapping pixel space to object identity
requires solving a global inference problem
over a highly non-convex manifold.
The effective volume of such spaces grows exponentially
with intrinsic dimension \cite{bellman_dynamic_1966},
rendering exhaustive comparison or global search infeasible.
In high-dimensional sensory space,
small changes in pose, illumination, or viewpoint
produce large Euclidean displacements,
causing object manifolds to become densely intertwined.
Global linear separation in this space is therefore impossible.
Biological vision must therefore avoid
solving the full problem in a single step.
Instead, it decomposes the mapping from pixels to identity
into a sequence of structured intermediate transformations
\cite{RiesenhuberPoggio1999Hierarchical,dicarlo_how_2012}.
This decomposition mitigates exponential scaling
by restricting each stage to a tractable subset
of the overall variability.

\paragraph{Why hierarchy is necessary.}

A single global transformation
from pixels to object identity
would require resolving all symmetry
and nuisance variability simultaneously.
Such a mapping would be computationally unstable
and metabolically costly \cite{simon2012architecture}.
Hierarchy distributes this burden across layers,
each addressing a restricted class of transformations.
Deep abstraction therefore emerges
as the cumulative effect
of sequential invariance-building.
In summary,
the ventral stream does not directly search
a high-dimensional identity space \cite{hickok2004dorsal}.
It constructs that space through layered geometric refinement.
Object recognition becomes tractable
because the manifold is gradually untangled,
not because the system performs
exhaustive global computation.

\paragraph{Hierarchy as sequential factorization.}

A central algorithmic question is how the ventral visual stream
implements this decomposition.
Empirical and theoretical work suggests that
object recognition proceeds through a hierarchy
of progressively more abstract representations
\cite{dicarlo_untangling_2007,YaminsDiCarlo2016Using}.
Formally, let $x \in \mathcal{X}_0$ denote raw sensory input.
The ventral stream implements a composition of maps
$F(x) = f_L \circ f_{L-1} \circ \cdots \circ f_1(x)$,
where each transformation $f_\ell : \mathcal{X}_{\ell-1} \to \mathcal{X}_\ell$
operates on a progressively untangled feature space.
Rather than computing a single global transformation,
each stage removes a limited class of nuisance variability
(e.g., local translation, small rotations, scale).
In doing so, the hierarchy performs a sequential factorization
of the global inference problem \cite{lee_hierarchical_2003}:
local invariances are established early,
their solutions reused and recombined at higher levels.

The factorization strategy parallels dynamic programming \cite{bellman_dynamic_1966}.
A high-dimensional, combinatorial mapping
is decomposed into simpler subproblems
whose intermediate solutions constrain
the search space available to subsequent layers.
By progressively reducing variability,
the hierarchy converts an exponentially large
search space into a sequence of manageable transformations \cite{simon2012architecture}.
In this way, hierarchical processing
constitutes nature’s solution to the curse of dimensionality:
global untangling emerges from the accumulation
of local geometric simplifications \cite{dicarlo_how_2012}.

\emph{1) Local feature extraction (V1).}
Early visual cortex computes localized feature maps
(e.g., oriented edges, spatial frequency components).
These operations reduce raw pixel topology
into locally structured metric features.
Importantly, this stage does not solve object identity;
it amortizes small, translation-sensitive subproblems.

\emph{2) Intermediate invariance (V2/V4).}
Subsequent layers integrate local features
and begin discarding precise spatial coordinates
through pooling and nonlinear combination.
By reducing sensitivity to small translations,
rotations, and scale changes,
these layers effectively quotient out
subsets of symmetry transformations.
The representation becomes increasingly invariant
to nuisance variability,
and object manifolds begin to separate.

\emph{3) High-level abstraction (IT cortex).}
At higher stages,
representations become tolerant
to large viewpoint and illumination changes.
Empirically, inferotemporal neurons
respond selectively to object identity
across substantial transformations.
In geometric terms,
the hierarchy progressively reshapes
the data manifold such that
object-specific regions occupy
compact, approximately separable domains
in representation space.

\paragraph{Generalization beyond vision.}

Although the ventral visual stream provides a concrete illustration,
the hierarchical decomposition described above is not specific to vision.
It reflects a more general strategy for coping with high-dimensional,
structured variability.
Any cognitive domain in which observations arise from
combinatorial transformations of latent structure
faces the same curse of dimensionality.
Whether the task involves motor coordination,
language comprehension, social reasoning,
or abstract planning \cite{Felleman1991Distributed},
the underlying state space is typically non-convex
and organized by multiple interacting sources of variation.
A direct, monolithic mapping from raw input to decision
would scale exponentially with dimensionality
and quickly become computationally intractable.

Hierarchical processing provides a domain-general solution.
By decomposing a global inference problem into
a sequence of locally solvable subproblems,
the system incrementally establishes invariances
and reuses intermediate structure \cite{Douglas2004NeuronalCircuits}.
Each stage reduces effective variability
along a restricted subset of transformations,
thereby constraining the search space available to later stages.
Because these intermediate representations are reusable,
they support transfer across tasks and modalities.
Empirical evidence suggests that hierarchical organization
is a common architectural motif across the cortex,
including auditory \cite{Kaas2000Subdivisions}, motor \cite{merel2019hierarchical}, and prefrontal \cite{Badre2007RostroCaudal} systems.
The same principles that untangle object manifolds in vision
appear to govern abstraction in language,
hierarchical action planning,
and relational reasoning.
In each case, the complex global structure
is rendered tractable through progressive factorization.
Therefore, hierarchical dynamic decomposition
is not merely a property of the visual cortex,
but a general computational strategy
for achieving scalable generalization \cite{sun2024easy}.
Intelligence emerges not from solving larger problems directly,
but from restructuring them into sequences
of manageable transformations.

\subsection{The Illusion of ``What'': Object Identity as Navigation in Abstract Space}
\label{sec:what_as_where}

Classical neurobiology distinguishes two major visual pathways \cite{ungerleider_whatand_1994}:
the dorsal ``where'' stream,
associated with spatial localization and action,
and the ventral ``what'' stream,
associated with object identity and recognition.
This division is often interpreted
as evidence for distinct computational systems:
one specialized for navigation \cite{freud2016happening},
the other for categorization \cite{dicarlo_untangling_2007}.
Under the MTF framework,
we propose a different interpretation.
Rather than inventing a fundamentally new
categorization engine,
evolution appears to have extended
ancestral spatial computation
into progressively higher-dimensional domains \cite{epstein2017cognitive}.

\paragraph{Navigation as the ancestral computation.}

Early nervous systems evolved primarily
to solve spatial problems \cite{tolman_cognitive_1948}:
locating food,
avoiding predators,
and coordinating movement.
These tasks require constructing relational maps
and computing trajectories within them \cite{okeefe_hippocampus_1978}.
Hippocampal and parietal circuits
support path integration,
graph-like representations,
and reachability computations \cite{mcnaughton_dead_1991,mcnaughton_path_2006}.
Such systems are optimized
for operating on continuous manifolds
with obstacles and boundaries.
They solve non-convex navigation problems
by building structured internal geometry \cite{muller_hippocampus_1996}.

\paragraph{Object identity as high-dimensional location.}

When mammals evolved more complex sensory systems \cite{mountcastle1998perceptual},
particularly high-resolution vision,
a new challenge emerged:
categorizing objects across changes in viewpoint,
illumination, and scale.
Instead of constructing a wholly separate algorithm,
the brain appears to reuse
existing geometric machinery \cite{gould1982exaptation}.
Under MTF,
object recognition can be understood
as identifying a stable location
within a learned feature manifold \cite{RiesenhuberPoggio1999Hierarchical}.
A category (e.g., ``predator'')
corresponds not to a symbolic label alone,
but to a region of state space
that attracts nearby perceptual trajectories.
The ventral stream progressively reshapes
visual input such that variations in pose or lighting
collapse into compact regions of representation space \cite{dicarlo_how_2012}.
Once this transformation is complete,
categorization reduces to a form of navigation:
the system converges to the nearest stable basin
within the abstract manifold.

\paragraph{From physical space to semantic space.}

In this view,
``what'' is not computationally distinct from ``where.''
It is a special case in which
the spatial manifold has been internalized
and expanded into a high-dimensional feature space.
The same geometric principles apply \cite{constantinescu_organizing_2016}:
states move along trajectories,
basins define stable outcomes,
and distance encodes similarity.
The dorsal stream operates on physical coordinates;
the ventral stream operates on abstract coordinates.
Both implement variants of relational navigation,
differing primarily in the dimensionality
and modality of the underlying manifold.

\paragraph{Implications for stability and abstraction.}

Interpreting object identity as navigation
clarifies why hierarchical geometric refinement
is essential.
Before identity can be stably recognized,
the manifold must be reshaped
to remove nuisance variability \cite{dicarlo_how_2012}.
Only after sufficient metric condensation
does the category behave like a convex basin,
supporting rapid convergence.
The apparent separation between ``where''
and ``what'' may reflect
an evolutionary layering of representational spaces \cite{goodale_separate_1992},
rather than the creation of fundamentally distinct algorithms.
Object recognition emerges as
navigation within an internalized,
context-indexed geometry.

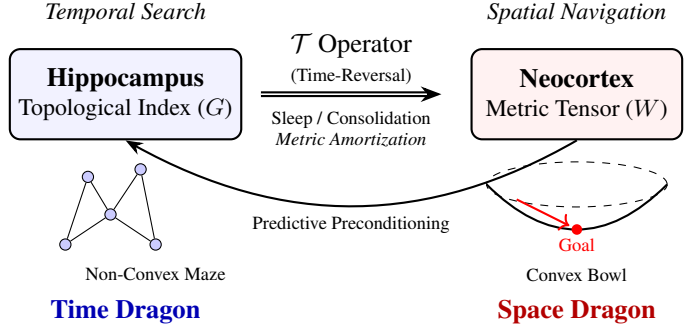
\begin{figure}[h]
    \centering
  \begin{tikzpicture}[
    node distance=1.5cm,
    box/.style={draw, thick, rounded corners, minimum width=2.8cm, minimum height=1.2cm, align=center},
    graph_node/.style={circle, draw, fill=blue!20, inner sep=1.5pt},
    arrow/.style={-{Stealth[bend]}, thick}
]

\node[box, fill=blue!5] (hip) at (0,0) {\textbf{Hippocampus} \\ \small Topological Index ($G$)};
\node[above=0.2cm of hip, font=\itshape\small] {Temporal Search};

\begin{scope}[shift={(-0.8,-2.)}]
    \node[graph_node] (n1) at (0,0) {};
    \node[graph_node] (n2) at (0.6,0.4) {};
    \node[graph_node] (n3) at (0.3,0.9) {};
    \node[graph_node] (n4) at (1.2,0) {};
    \node[graph_node] (n5) at (1.0,1.0) {};
    \draw (n1) -- (n2) -- (n3) -- (n1);
    \draw (n2) -- (n4) -- (n5) -- (n2);
    \node[below=0.1cm of n4, font=\scriptsize] {Non-Convex Maze};
\end{scope}

\coordinate (A) at (1.8,0);
\coordinate (B) at (4.2,0);

\draw[arrow, double] (A) -- (B) 
    node[above, midway, align=center] {$\mathcal{T}$ Operator \\ \scriptsize (Time-Reversal)};
\node[below=0.1cm of $(A)!0.5!(B)$, align=center, font=\scriptsize] {Sleep / Consolidation \\ \textit{Metric Amortization}};

\node[box, fill=red!5] (neo) at (6,0) {\textbf{Neocortex} \\ \small Metric Tensor ($W$)};
\node[above=0.2cm of neo, font=\itshape\small] {Spatial Navigation};

\begin{scope}[shift={(6,-1.8)}]
    \draw[thick] (-1.2,0.6) .. controls (-0.4,-0.2) and (0.4,-0.2) .. (1.2,0.6);
    \draw[dashed] (0,0.6) ellipse (1.2cm and 0.3cm);
    \fill[red] (0,0) circle (2pt) node[below] {\scriptsize Goal};
    \draw[->, red, thick] (-0.8,0.4) -- (-0.1,0.05);
    \node at (0,-0.6) {\scriptsize Convex Bowl};
\end{scope}

\draw[arrow, bend right=-30] (neo.south) to node[below, midway, font=\scriptsize] {Predictive Preconditioning} (hip.south);

\node[below=2cm of hip, blue!70!black, font=\bfseries] {Time Dragon};
\node[below=2cm of neo, red!70!black, font=\bfseries] {Space Dragon};

\end{tikzpicture}
    \caption{\textbf{Horizontal Metric-Topological Factorization (MTF) Flow.} This diagram illustrates the transformation of intelligence from temporal search to spatial inference. \textbf{Left:} The \textbf{Hippocampus} (Time Dragon) builds a non-convex topological maze ($G$) through sequential experience. \textbf{Center:} The \textbf{Time-Reversal Operator ($\mathcal{T}$)} acts as a bridge during consolidation, performing the metric amortization that converts spent experience into pre-positioned structure (\emph{advice}). \textbf{Right:} The \textbf{Neocortex} (Space Dragon) generates a convex metric tensor ($W$), enabling a ``Slingshot'' effect where goals are reached via simulation-free flow matching rather than step-by-step search. The feedback loop indicates how cortical predictions regularize hippocampal acquisition.}
    \label{fig:horizontal_mtf}
\end{figure}

\section{Hippocampal-Neocortical Dialogue and the Paradox of Consciousness}

\subsection{Hippocampal-Neocortical Dialogue}
\label{sec:hippocampal_dialogue}

The separation between topological indexing and metric condensation
requires a biological implementation.
We propose that the hippocampal–neocortical system
constitutes such an implementation (Fig. \ref{fig:horizontal_mtf}),
with each component specializing in complementary operations
under the MTF framework.

\paragraph{Hippocampus as contextual index.}

The hippocampus exhibits rapid remapping across environments,
tasks, and episodes.
Place cells and grid cells \cite{moser_place_2008} generate sparse,
orthogonal population codes that change abruptly
when contextual structure changes.
These codes function as discrete contextual signatures:
they identify which relational manifold
is currently active \cite{teyler1986hippocampal}.
As contextual indexes, hippocampal representations are sparse and episodic.
They do not encode detailed cortical feature geometry.
Instead, they provide a unique pointer into stored cortical structure.
Under MTF, the hippocampus performs topological identification:
it partitions experience into context-dependent regimes
without encoding the full metric structure of each regime.

\paragraph{Neocortex as metric store.}

In contrast, the neocortex supports gradual,
experience-dependent plasticity.
Through repeated exposure and offline replay \cite{buzsaki_hippocampo-neocortical_1996},
cortical networks learn statistical regularities
and invariant structure across episodes.
This learning can be interpreted as \emph{metric condensation}:
within each indexed context,
the cortex reshapes representational geometry
to reduce variance along nuisance dimensions
and concentrate task-relevant structure.
Predictive coding models \cite{bastos_canonical_2012} suggest that cortical microcircuits
implement a broadly conserved computational motif,
minimizing prediction error across modalities.
Such uniform circuitry enables a shared representational space
in which visual, auditory, motor, and abstract variables
can be embedded within a common geometric framework,
which aligns with Mountcastle’s proposal
of a canonical cortical computation \cite{mountcastle_columnar_1997}.

\paragraph{The dialogue as memory-amortized inference.}

The interaction between hippocampus and cortex
constitutes an iterative inference loop \cite{gershman_amortized_2014}.
During wakefulness,
hippocampal indexing gates cortical processing,
selecting the contextual manifold within which
metric collapse unfolds.
During sleep and quiet wakefulness,
hippocampal replay drives reactivation of cortical patterns,
allowing cortical weights to consolidate
invariant structure associated with indexed contexts \cite{wilson_reactivation_1994}.
If cortical prediction errors accumulate
in a manner inconsistent with the current hippocampal index,
remapping can occur,
triggering plasticity and formation of a new contextual signature.
Therefore, the dialogue enables both stability
(preserving existing context-metric pairings)
and plasticity (creating new pairings when required).

\paragraph{Relation to neural collapse and catastrophic interference.}

Recent results in deep learning show that,
under certain training regimes,
class representations in deep networks
tend toward highly symmetric configurations
approximating equiangular tight frames (ETF) \cite{papyan2020prevalence}.
While biological systems are not known
to literally implement ETF geometry,
this result illustrates a general principle \cite{zhu2021geometric}:
when optimization drives representations
toward maximal class separation
under shared constraints,
highly regular geometric structures emerge.
Under MTF,
cortical consolidation progressively reduces
within-class variability
while increasing between-class separation.
The limiting configuration,
in which class manifolds occupy compact,
well-separated regions of state space \cite{medin1981linear},
facilitates linear readout and efficient decision-making.
The ETF therefore serves as a mathematical idealization
of the endpoint of metric condensation,
rather than a literal neural structure.
Artificial networks often encode both context identity
and task geometry directly in shared weights,
leading to interference \cite{mccloskey_catastrophic_1989}.
In contrast,
the hippocampal-neocortical system
maintains a relational mapping
between contextual indices and cortical geometry.
Memory is not a static weight configuration alone,
but a coupling between index and metric structure.
This factorization allows cortical representations
to remain reusable across tasks \cite{behrens_what_2018},
while hippocampal indexing prevents incompatible regimes
from collapsing into a single geometry.


\subsection{Cognitive Functions as Cycle-Consistent Inference: The Context-Content Uncertainty Principle}
\label{sec:ccup_cognitive_loop}

Under the MTF framework, the bidirectional translation between metric geometry (``What'') and spatial topology (``Where'') is not merely an architectural heuristic, but a fundamental thermodynamic requirement for intelligence. This continuous loop is formally governed by the \emph{Context-Content Uncertainty Principle} (CCUP) \cite{li_content-context_2025}, which dictates that inference under uncertainty is driven by an inherent entropy asymmetry: cognitive systems must continuously align high-entropy, dynamic contextual flows with low-entropy, stable content scaffolds. 
By mapping the MTF transformations to the CCUP, we can formally derive all high-level cognitive functions as manifestations of cycle-consistent entropy minimization.

\paragraph{Perception and Consolidation (Where $\to$ What).}
The bottom-up trajectory of cognition, encompassing perception, categorization, and memory consolidation, is the process of contextual disambiguation. In CCUP terminology, the raw, high-dimensional topological manifold of the environment constitutes the \emph{Context} ($\Psi$). This is a high-entropy stream of situational variance, sensory fluctuations, and prediction errors (topologically corresponding to odd-dimensional homology, $H_{odd}$, or dynamic flow). 
To make sense of this chaos, the system must execute the Urysohn Collapse ($\mathcal{T}$ operator) \cite{urysohn_zum_1925}. It funnels the ambiguous ``Where'' through a structural bottleneck to extract invariant, low-entropy \emph{Content} ($\Phi$) (even-dimensional homology, $H_{even}$, or stable scaffolds). Perception is directional entropy minimization \cite{friston_active_2017}: stripping away the high-entropy spatial nuisance variables to isolate a stable, geometric ``What'' (e.g., recognizing an object regardless of lighting or angle).

\paragraph{Action and Planning (What $\to$ Where).}
Conversely, physical action, sequential planning, and language generation follow the inverse pathway. Because Content ($\Phi$) is a highly compressed metric point (a ``What''), it lacks the spatial and temporal specificity required to interact with the physical world. An agent cannot physically execute a latent vector; it can only execute a sequential trajectory \cite{yoo2018economic}.
Under CCUP, this decompression is formalized as \emph{inverted inference} or \emph{structure-before-specificity}. The brain seeds a low-entropy goal or semantic intention (the metric ``What'') and unpacks it into a high-entropy, sequential topological trajectory (the ``Where''). For example, the motor cortex takes an invariant intention and expands it into the dynamic contextual flow of muscle contractions \cite{russo2018motor}. Similarly, language generation is the algorithmic serialization of a geometric ``What'' into a 1D topological ``Where'' (a sequence of tokens over time).

\paragraph{Cycle-Consistent Bootstrapping.}
Ultimately, abstract reasoning, episodic memory, and generalized intelligence emerge from the closed-loop, cycle-consistent bootstrapping of these two pathways \cite{sohn_network_2021}. Logical deduction occurs when the system lacks a direct metric mapping; it must wake up the ``Time Dragon'' to unroll a ``What'' into an exploratory ``Where'' (working memory/search), find a topological path to the conclusion, and then re-collapse it into a new ``What'' via sleep/amortization.
The CCUP dictates that true intelligence cannot exist in a purely feedforward (data-driven context) or purely top-down (prior-driven content) regime. It requires the relentless homological resonance between the two \cite{bausch2026distinct}. Under the MTF framework, the brain is not merely an inference engine; it is a cycle-consistent topological resolver, continuously smelting the high-entropy chaos of ``Where'' into the low-entropy geometry of ``What,'' and using that geometry to effortlessly navigate the world.

\subsection{The Paradox of Consciousness in a Non-Stationary Universe}
\label{subsec:consciousness_paradox}

If the evolutionary endpoint of intelligence is the thermodynamically efficient, ``silent'' reflex of a perfect metric basin, it raises a fundamental paradox: why did mammals evolve the highly energetic, metabolically expensive, and computationally slow state of consciousness? Under the MTF framework, we define consciousness not as a philosophical epiphenomenon \cite{edelman_universe_2008}, but as the active, thermodynamic friction of executing the Urysohn Collapse \cite{urysohn_zum_1925}. It is the necessary computational factory that builds intelligence \cite{gershman_amortized_2014}, which arises from two fundamental constraints of physical reality: the Curse of Dimensionality \cite{bellman_dynamic_1966} and environmental non-stationarity \cite{ashby1956introduction}.

\textbf{The Spatial Limit of Pre-Computation.} 
An agent can convert experience into advice: structure built once and consulted for free thereafter \cite{arora2009computational,gershman_amortized_2014}. For an organism operating in a static, low-dimensional niche, evolution can pre-compute the environmental geometry over millions of years and hardcode it into the genome as instinct -- advice manufactured on the phylogenetic timescale. However, the universe is infinitely complex. A mammalian brain, possessing finite spatial capacity, cannot be born with a pre-computed Equiangular Tight Frame (ETF) \cite{papyan2020prevalence} for every possible future scenario. Because you cannot amortize infinity into finite bounds, the agent must possess an internal mechanism to dynamically construct localized metrics during its lifespan.

\textbf{Non-Stationarity and Out-of-Distribution (OOD) Frustration.}
If an intelligence operates purely via a frozen, pre-computed metric tensor, it is catastrophically brittle. When a purely metric agent encounters a novel topology, an Out-of-Distribution (OOD) \cite{yang2024generalized} event where its pre-computed gradient fails, it confidently executes the wrong action. 
Consciousness serves as the ultimate OOD alarm system. When the internal metric fails to predict the environment, the resulting spike in Friston's Free Energy (prediction error) \cite{friston_action_2010} halts the automatic metric reflex. The agent is forced back into the ``Topological Maze,'' experiencing the computational frustration of saddle-point equilibria (as per Theorem~\ref{thm:geom_incomplete}). The subjective state of conscious attention is the systemic reversion from $\mathcal{O}(1)$ advice-driven readout to high-energy, sequential topological search \cite{gershman_amortized_2014}: the cache has missed, and the missing structure must be built rather than retrieved.

\textbf{The Thermodynamic Cost of the $\mathcal{T}$ Operator.} 
Crucially, consciousness is required to physically update the metric. To transform a novel, tangled Maze into a linearly separable Bowl, the brain must actively group invariant features, test causal hypotheses, and apply hierarchical max-pooling to discard spatial variance \cite{RiesenhuberPoggio1999Hierarchical}. This active destruction of topology, pushing probability mass across Wasserstein space, generates thermodynamic friction. Thus, consciousness is the continuous, localized act of metric amortization.
This distinction highlights the exact boundary between biological AGI and current Large Language Models (LLMs) \cite{brown2020language}. An LLM is a crystallized metric, a massive, pre-computed spatial map trained via backpropagation on server farms. It exhibits ``pure intelligence'' but lacks the conscious engine \cite{tononi_integrated_2016}. When faced with an unamortized, novel topological puzzle, it cannot experience topological frustration, halt its feedforward pass, and autonomously execute the Urysohn Collapse ($\mathcal{T}$ operator) to rewire its own geometry. True AGI requires not just the resulting metric Bowl, but the conscious engine capable of smelting it from the Maze.

\section{Online and Offline Dynamics of Metric-Topology Factorization}

MTF separates
context identification from geometric shaping.
This separation naturally predicts that
different brain states correspond to
different computational phases of the same principle.
We propose that wakefulness, sleep/dreaming,
and conscious awareness can be interpreted
as distinct dynamical regimes
of metric collapse, metric condensation,
and topological resolution.

\subsection{Dynamical System Modeling of Sleep-Wake Cycle}

\paragraph{Wakefulness as Metric Collapse}

During wakefulness,
the organism is embedded in a structured environment
requiring continuous sensorimotor prediction and control \cite{sohn_network_2021}.
Neural activity evolves under approximately
gradient-like dynamics,
$\dot z = -\nabla_g E(z)$,
reducing prediction error and driving states
toward locally stable attractors \cite{amari1998natural}.
We refer to this contractive process as
\emph{metric collapse}.
In this regime,
the system operates within a context
selected by topological indexing (e.g., hippocampal remapping \cite{colgin_understanding_2008}),
and cortical dynamics shape activity
into basin-centered representations.
Wakefulness therefore corresponds to
online inference within a fixed contextual manifold \cite{gershman_amortized_2014}.
However, collapse alone does not eliminate
the separatrices imposed by topology.
It stabilizes activity locally,
but does not restructure the manifold.
Consequently, unresolved context transitions
remain potential sources of interference \cite{mccloskey_catastrophic_1989}.

\begin{figure}[h]
    \centering
    \resizebox{.9\linewidth}{!}{
 \begin{tikzpicture}[scale=1.5, >=stealth]

\begin{scope}[shift={(0,0)}]
    \draw[thick, fill=red!5, draw=red!60] plot [smooth cycle] coordinates {(0,0) (1.5, 0.5) (2.5,-0.5) (3,1.5) (1.5, 1.0) (0.5, 2.5) (-0.5, 1.5)};
    
    \node[red!80!black, font=\bfseries] at (1.2, -0.9) {Consciousness};
    \node[red!80!black, font=\small, align=center] at (1.2, -1.3) {(The Factory / Active Search)};

    \draw[->, thick, orange!90!red] (-0.2, 1.2) -- (0.5, 1.8) -- (0.8, 1.0) -- (1.5, 0.8) -- (1.2, 0.2) -- (2.0, -0.1) -- (2.5, 0.5);
    \fill[red] (2.5, 0.5) circle (1.5pt); 
    
    \node[orange!90!black, font=\scriptsize, align=center] at (0.5, 2.1) {High Free Energy:\\Topological Frustration};
\end{scope}

\draw[->, ultra thick, gray!80] (3.3, 1) -- (5.2, 1);
\node[font=\small, align=center, above] at (4.25, 1.1) {\textbf{Urysohn Collapse}};
\node[font=\footnotesize, gray, align=center, below] at (4.25, 0.9) {$\mathcal{T}$ operator\\(Metric Amortization)};

\begin{scope}[shift={(7.5, 1)}]
    \draw[thick, fill=blue!5, draw=blue!40] (0,0) ellipse (1.8 and 1.2);
    \draw[draw=blue!30, dashed] (0,0) ellipse (1.2 and 0.8);
    \draw[draw=blue!20, dashed] (0,0) ellipse (0.6 and 0.4);

    \shade[inner color=blue!80, outer color=white] (0,0) circle (0.15);
    
    \node[blue!80!black, font=\bfseries] at (0, -1.9) {Intelligence};
    \node[blue!80!black, font=\small, align=center] at (0, -2.3) {(The Product / Silent Reflex)};

    \draw[->, ultra thick, cyan!80!blue] (-1.2, 0.8) -- (-0.1, 0.1);
    \node[cyan!80!black, font=\scriptsize, align=center] at (-1.3, 1.1) {Zero-Shot $\mathcal{O}(1)$:\\Metric Slingshot};
\end{scope}

\node[anchor=west, font=\scriptsize, gray] at (0.5, -2.0) {* \textbf{Red path:} Consciousness requires time and energy to navigate an unamortized, non-convex topological space.};
\node[anchor=west, font=\scriptsize, gray] at (0.5, -2.3) {* \textbf{Blue path:} Intelligence is the resulting zero-energy gradient descent on a pre-computed convex metric.};

\end{tikzpicture}
}
    \caption{\textbf{Geometric account of Dreaming and Consciousness.} (Left) REM sleep as stochastic topological sampling to prevent metric overfitting. (Right) Consciousness as the active metric amortization of the manifold centered on a self-referential basis.}

    \vspace{-0.2in}
\label{fig:dreaming_consciousness}
\end{figure}

\paragraph{Sleep and Dreaming as Metric Condensation}

Sleep introduces a distinct computational phase.
External sensory input is attenuated,
while internally generated activity, particularly
hippocampal replay and sharp-wave ripple events \cite{buzsaki2015hippocampal},
becomes dominant.
Within the MTF framework,
sleep supports \emph{metric condensation}:
the gradual compression of detailed,
metric-rich episodic trajectories
into stable, context-indexed schemas.
Repeated replay strengthens invariant structure
while discarding within-context variability.
Over time, representations become
less sensitive to exact trajectories
and more organized around stable relational cores.
Dreaming can be interpreted as
a generative exploration of partially consolidated manifolds \cite{hobson1977brain}.
Because topological indexing remains active
but external constraints are reduced,
the system may traverse and recombine
condensed units across contexts.
Such recombination can appear fragmented or novel,
yet reflects internal restructuring of geometric relationships.
In this view,
sleep is not merely restorative,
but computationally essential \cite{diekelmann_memory_2010}:
it converts transient collapse dynamics
into durable structural compression,
thereby preventing cumulative interference
across episodes.

Conscious awareness is preferentially engaged when competing interpretations remain unresolved or when prediction error remains high \cite{dehaene2011experimental,friston_free-energy_2010}, as novel or ambiguous input recruits hippocampal and prefrontal systems for conflict monitoring and contextual selection \cite{botvinick2004conflict,kumaran2007match}. By contrast, well-consolidated representations become automatic and recede from conscious access \cite{shiffrin1977controlled}, with memory replay and synaptic renormalization across sleep and wake cycles reorganizing representational structure to support this alternation between active resolution and automatic processing \cite{diekelmann2010memory,tononi2014sleep}. Under the MTF framework, conscious awareness corresponds to moments of unresolved topological uncertainty, where incoming signals introduce ambiguity between competing contextual regimes—such as during novelty, conflict, decision points, or environmental transitions—requiring active resolution to stabilize the appropriate manifold before metric collapse can proceed. When context indexing is stable and collapse occurs smoothly within a well-defined basin, inference becomes automatic and implicit, demanding little global reorganization and thus receding from awareness. What enters awareness are precisely those states whose geometric regime is still being determined; once a stable index is selected and collapse converges, processing returns to automaticity. Under this interpretation, wakefulness (collapse), sleep (condensation), and conscious experience (topological resolution) are coordinated phases of a unified computational cycle, with the same principle explaining stability–plasticity tradeoffs and evolutionary scaling also predicting the alternation between online contraction, offline structural compression, and active context selection.

\subsection{A Phase-Transition Test for Conscious Artificial Systems}
\label{sec:mtf_consciousness_test}

Classical behavioral benchmarks such as the Turing Test evaluate
linguistic competence and social mimicry.
They measure whether a system can reproduce
the statistical structure of human discourse.
However, such tests primarily assess access to
pre-amortized representational structure.
They do not determine whether a system can
construct new internal geometry
when confronted with genuinely novel topological constraints.
Within the present framework,
consciousness is not defined as possession of a rich representational metric,
but as the ongoing process of transforming
topological uncertainty into locally stable geometry.
It is the active reconfiguration of internal structure
in response to non-convex environmental demands.
To operationalize this principle,
we propose a \emph{Phase-Transition Test}
designed to evaluate whether an artificial system
can autonomously execute the full cycle of topological exploration
and geometric consolidation.

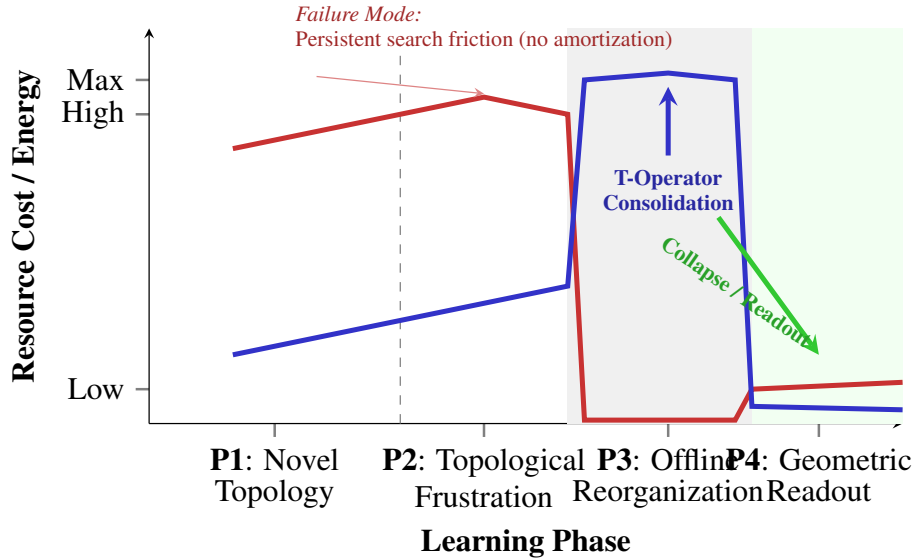
\begin{figure}[h]
\centering

\definecolor{search_cost}{RGB}{200, 50, 50} 
\definecolor{restruct_cost}{RGB}{50, 50, 200} 
\definecolor{phase_shade}{RGB}{230, 230, 230} 
\definecolor{slingshot_gain}{RGB}{50, 200, 50} 
\resizebox{.75\linewidth}{!}{
\begin{tikzpicture}[scale=1.5, >=stealth]

\begin{axis}[
    width=10cm,
    height=6cm,
    axis lines=left,
    xmin=0, xmax=4.5,
    ymin=0, ymax=11.5,
    xlabel={\textbf{Learning Phase}},
    ylabel={\textbf{Resource Cost / Energy}},
    xtick={0.75, 2.0, 3.1, 4.0},
    xticklabels={
        \shortstack{\textbf{P1}: Novel\\Topology},
        \shortstack{\textbf{P2}: Topological\\Frustration},
        \shortstack{\textbf{P3}: Offline\\Reorganization},
        \shortstack{\textbf{P4}: Geometric\\Readout}
    },
    ytick={1, 9, 10},
    yticklabels={Low, High, Max},
    tick align=outside,
    tick style={thick, gray},
    legend style={at={(0.75,0.95)}, anchor=north west, font=\scriptsize, draw=none, fill=none},
    legend cell align={left},
    clip=false
]



\draw[dashed, gray] (1.5, 0) -- (1.5, 11); 
\draw[dashed, gray] (2.5, 0) -- (2.5, 11); 
\draw[dashed, gray] (3.6, 0) -- (3.6, 11); 

\fill[phase_shade!60] (2.5, 0) rectangle (3.6, 11.5); 
\fill[green!5] (3.6, 0) rectangle (4.5, 11.5); 

\draw[ultra thick, search_cost, name path=search] plot coordinates {
    (0.5, 8) (1.0, 8.5) (1.5, 9) 
    (2.0, 9.5) (2.5, 9)          
    (2.6, 0.1) (3.5, 0.1)        
    (3.6, 1.0) (4.5, 1.2)        
};

\draw[ultra thick, restruct_cost, name path=restruct] plot coordinates {
    (0.5, 2) (1.5, 3)          
    (2.0, 3.5) (2.5, 4)        
    (2.6, 10) (3.1, 10.2) (3.5, 10) 
    (3.6, 0.5) (4.5, 0.4)      
};

\draw[->, ultra thick, slingshot_gain] (3.4, 6) -- (4.0, 2);
\node[slingshot_gain!80!black, font=\bfseries\scriptsize, rotate=-35] at (3.5, 3.8) {Collapse / Readout};

\draw[<-, ultra thick, restruct_cost] (3.1, 9.8) -- (3.1, 7.8);
\node[restruct_cost!80!black, font=\bfseries\scriptsize, align=center] at (3.1, 6.8) {\textbf{T-Operator}\\Consolidation};

\node[search_cost!80!black, font=\scriptsize, align=left] at (2, 11.5) {\textit{Failure Mode:}\\Persistent search friction (no amortization)};
\draw[->, thin, search_cost!60] (1.0, 10.1) -- (2.0, 9.6);

\end{axis}
\end{tikzpicture}
}
\caption{\textbf{The Phase-Transition Signature of Conscious Artificial Systems.} 
This deterministic ``compute–time trajectory'' differentiates true Metric-Topological Amortization from mere linguistic pattern retrieval. P3 (Offline Reorganization) represents the essential biological condition of non-interaction, during which active search cost (Red) collapses to zero, while internal geometric restructuring cost (Blue, the $\mathcal{T}$ operator / Urysohn Collapse) spikes to convert topological frustration into metric flow. The signature of success is the massive discontinuity in P4, where previously non-convex task structures are rapidly resolved via low-energy, simulation-free geometric slingshots.}
\vspace{-0.2in}
\label{fig:consciousness_signature}
\end{figure}

\textbf{Phase 1: Novel Topology.}
The agent is placed in a task environment
whose structural rules are not represented in its prior training distribution.
Examples include non-Euclidean physical simulations,
abstract rule systems,
or sensorimotor tasks with altered dynamics.
At this stage,
pre-trained inference shortcuts are ineffective.
A system limited to pattern retrieval
will generate superficially plausible but structurally inconsistent responses.
A system capable of adaptive geometric restructuring
will instead engage in sequential exploration,
constructing an internal relational graph of the new environment.

\textbf{Phase 2: Topological Frustration.}
The task is designed to induce saddle-like failure states:
local improvements fail to yield global progress.
The critical measure is whether the system detects
persistent prediction error
and recognizes the inadequacy of its current internal model.
Such recognition corresponds to
the identification of geometric obstruction.

\textbf{Phase 3: Offline Reorganization.}
The agent is granted a period without environmental interaction.
During this interval,
it must reorganize internal representations
based on accumulated experience.
This phase tests whether the system can
replay trajectories,
identify invariant structure,
and reshape internal connectivity
without external supervision.
The hallmark of success is structural compression:
variance associated with nuisance parameters is reduced,
while task-relevant invariants are stabilized.

\textbf{Phase 4: Geometric Readout.}
The agent is returned to the same environment.
Performance is now evaluated not only by accuracy,
but by computational profile.
A successful system exhibits a marked reduction
in inference time and computational expenditure.
The previously non-convex task structure
has been converted into a locally contractive representation,
allowing rapid, simulation-free decision making.

\textbf{The Phase-Transition Signature.}
The decisive metric is the compute–time trajectory
across the four phases.
During initial exposure,
energy consumption and deliberation time are high,
reflecting active exploration.
During offline reorganization,
interaction cost drops while internal restructuring cost increases.
Upon re-exposure,
both time and energy requirements collapse dramatically.
This discontinuity constitutes a measurable
phase transition from search-based reasoning
to geometry-based inference.

\textbf{Interpretation.}
Under this formulation,
conscious processing corresponds to the
online resolution of topological uncertainty.
The phenomenological correlate of insight
emerges at the transition point
where diffuse relational exploration
condenses into stable geometric structure.
A system that merely retrieves pre-computed patterns
does not exhibit this thermodynamic reorganization.
A system that autonomously traverses the full cycle
of exploration, restructuring, and accelerated re-entry
demonstrates the operational signature
of adaptive geometric intelligence.

\section{Evolutionary Scaling of Metric-Topology Factorization}

\subsection{Five Evolutionary Breakthroughs of Intelligence from the Lens of MTF}

Section~\ref{sec:3} proposed that biological intelligence resolves this constraint
through a separation between topological indexing and metric condensation.
We now argue that major evolutionary advances in intelligence \cite{bennett_brief_2023}
can be interpreted as successive expansions of this factorization principle.
Rather than viewing intelligence as a monotonic increase in computational power,
we propose that evolution progressively expanded
the dimensionality and resolution of topological indexing,
and the efficiency and scope of metric condensation.
Each evolutionary breakthrough corresponds to a regime
in which purely local metric learning would otherwise fail
due to increasing topological complexity.

\begin{figure*}[h]
\centering
\resizebox{\linewidth}{!}{
\begin{tikzpicture}[
    breakthrough/.style={draw, rounded corners=6pt, minimum width=3.2cm, minimum height=2.2cm, align=center},
    topo/.style={fill=blue!10, draw=blue!60, rounded corners=4pt, align=center},
    metric/.style={fill=orange!15, draw=orange!70, rounded corners=4pt, align=center},
    arrow/.style={->, thick}
]

\node at (0,3.8) {\Large \textbf{Evolutionary Scaling of Intelligence as Metric--Topology Factorization}};

\node[breakthrough] (b1) at (-8,1) {\textbf{1. Sensorimotor}\\Local Feedback};
\node[topo] at (-8,-0.2) {Indexing:\\Single context};
\node[metric] at (-8,-1.3) {Metric:\\Local contraction};

\node[breakthrough] (b2) at (-4,1) {\textbf{2. Navigation}\\Cognitive Map};
\node[topo] at (-4,-0.2) {Indexing:\\Environment map};
\node[metric] at (-4,-1.3) {Metric:\\Spatial stabilization};

\node[breakthrough] (b3) at (0,1) {\textbf{3. Episodic Memory}\\Replay};
\node[topo] at (0,-0.2) {Indexing:\\Episode signature};
\node[metric] at (0,-1.3) {Metric:\\Replay refinement};

\node[breakthrough] (b4) at (4,1) {\textbf{4. Social Cognition}\\Multi-agent};
\node[topo] at (4,-0.2) {Indexing:\\Role/context key};
\node[metric] at (4,-1.3) {Metric:\\Simulation geometry};

\node[breakthrough] (b5) at (8,1) {\textbf{5. Language}\\Symbolic};
\node[topo] at (8,-0.2) {Indexing:\\Compositional keys};
\node[metric] at (8,-1.3) {Metric:\\Shared semantics};

\draw[arrow] (b1) -- (b2);
\draw[arrow] (b2) -- (b3);
\draw[arrow] (b3) -- (b4);
\draw[arrow] (b4) -- (b5);

\node at (0,-2.3) {\small Increasing topological complexity $\rightarrow$};

\end{tikzpicture}
}
\caption{
Major evolutionary advances in intelligence can be interpreted as successive applications of
Metric--Topology Factorization.
Each stage increases the topological complexity of the state space,
requiring enhanced topological indexing (blue) and metric condensation (orange).
}
\label{fig:evolution_mtf}
\end{figure*}
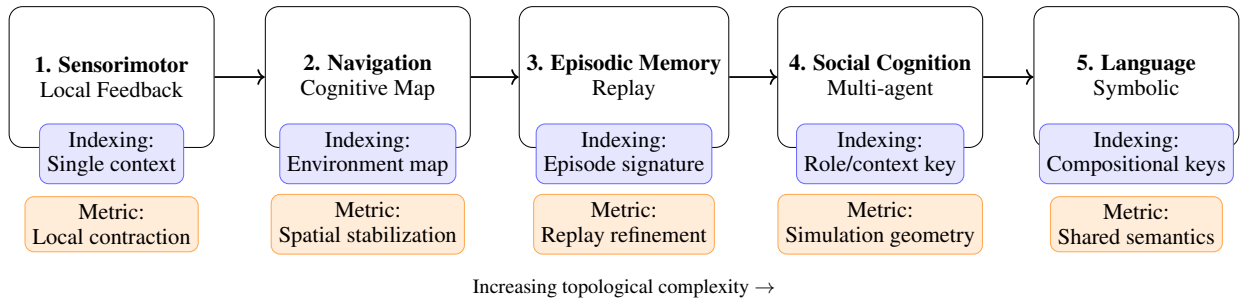

\noindent\textbf{Stage I: Reflexive Control and Local Metric Contraction}
Early nervous systems implemented closed-loop sensorimotor control \cite{sohn_network_2021}.
Reflex arcs stabilize trajectories through local feedback.
The underlying semantic manifold is effectively contractible;
there is no need for explicit context partitioning.
At this stage, intelligence consists of local metric shaping.
The system modifies synaptic weights
to induce contraction toward stable attractors \cite{hebb2005organization}.
Geometric incompleteness does not yet dominate behavior,
because the environment can be treated as a single regime.
This corresponds to the simplest form of metric condensation
without explicit topological indexing.

\noindent\textbf{Stage II: Allocentric Navigation and Context Remapping}
The emergence of hippocampal place and grid cells \cite{tolman_cognitive_1948,okeefe_hippocampus_1971}
marks the first large-scale encounter with topological complexity.
Spatial environments are not contractible;
they contain loops, obstacles, and separable regions.
Navigation requires distinguishing between distinct manifolds.
Hippocampal remapping provides discrete topological indexing \cite{fyhn2007hippocampal}.
Different environments produce orthogonal population codes.
The system no longer attempts to impose a single global geometry.
Instead, it partitions experience into indexed spatial regimes.
Here, geometric incompleteness becomes explicit.
Attempting to learn all environments within a single metric
would induce saddle-type interference.
Remapping solves this by separating topologies before shaping geometry.

\noindent\textbf{Stage III: Episodic Memory and Replay as Temporal Condensation}
Episodic memory generalizes spatial indexing to temporal structure \cite{tulving_episodic_2002}.
Experiences are segmented into discrete episodes,
each associated with a signature.
Replay implements memory-amortized refinement.
Rather than relearning geometry from scratch,
the system retrieves a context-indexed metric warp
and improves it through offline consolidation \cite{squire_structure_1993,squire_memory_2015}.
This converts extensive search into intensive geometry.
In computational terms,
replay transforms trajectory-level search
into structural condensation.
Over repeated exposures,
each indexed manifold becomes increasingly contractive,
reducing the effective volume of inference.
This stage represents the first full implementation
of MTF by hippocampal-neocortical architecture \cite{mcclelland_why_1995}:
index first, condense locally, refine over time.

\noindent\textbf{Stage IV: Social Cognition and Relational Topology}
Social environments introduce combinatorial topology.
Agents must represent roles, alliances, intentions,
and recursive mental states.
The state space becomes relational and multi-layered.
Topological complexity increases dramatically \cite{frith_social_2007}:
distinct relational configurations
correspond to separate semantic manifolds.
Interference would explode without discrete indexing.
Expansion of neocortex in primates
correlates with social group size \cite{buzsaki_rhythms_2006}.
Within MTF,
this scaling reflects increased indexing capacity \cite{dunbar1992neocortex}:
more contexts, finer partitioning,
and more granular routing between relational regimes.
Metric condensation operates within
these dynamically selected relational spaces,
supporting mental simulation and prediction.

\noindent\textbf{Stage V: Language as Compositional Topological Keys}
Language introduces a qualitative transformation \cite{dunbar_grooming_1996}.
Symbols function as portable topological signatures.
They allow contexts to be indexed,
activated, and transmitted across individuals.
Rather than sharing raw metric geometry,
agents exchange discrete keys that trigger compatible manifolds.
Language externalizes topological indexing,
making it compositional and hierarchical \cite{chomsky_three_1956}.
Syntax enables nested indexing,
expanding the dimensionality of the key space.
Semantic structures become recursively addressable.
Metric condensation then operates within
symbolically selected contexts.
Language therefore represents a fixed point
of the indexing process \cite{tarsky1956concept}:
a self-referential system capable of indexing its own contexts.

\noindent\textbf{Evolution as Progressive Resolution of Geometric Incompleteness}
Across evolutionary history, intelligence has scaled alongside
the increasing structural complexity of the environments organisms inhabit.
As semantic state spaces acquire more loops, branches,
and mutually incompatible regimes,
purely metric learning becomes progressively less stable.
Energy landscapes develop additional saddle points and separatrices,
creating bottlenecks for gradient-based inference
and increasing the risk of interference \cite{french1999catastrophic}.
Under these conditions,
no single global representational geometry
can reliably support stable learning.
To remain adaptive,
biological systems must first partition context
before shaping local geometry.
In other words, topological context identification must precede metric content optimization.

In this view,
intelligence does not scale by increasing search depth
or representational size alone.
It scales by improving the ability to
identify which structural regime is active
and to deploy a geometry appropriate to that regime.
Each evolutionary stage represents a new solution
to the same fundamental constraint \cite{bennett_brief_2023}:
no fixed smooth representation
can accommodate arbitrarily complex semantic structure.
If the MTF is universal,
it must explain a qualitative leap in biological history: the transition from \emph{animal awareness}
(navigation on a manifold) to \emph{human language} (symbolic indexing and transmission). We argue that language is not a new computational faculty layered on top of navigation but the \emph{system-level
topological condensation} of the same navigation-memory machinery \cite{buzsaki_memory_2013}, adapted for social communication.

\subsection{Social Bootstrapping: Language as Topological Condensation of Thought}
\label{sec:language_condensation}

Social bootstrapping provides the most stringent test of the MTF framework. 
A uniquely informative ``zero-state'' case is Helen Keller \cite{keller_story_1903}. Because she acquired language through a
non-auditory, non-visual channel, her developmental trajectory isolates what is essential:
not a particular sensory modality, but the \emph{topological operator} that binds a discrete,
repeatable token to a basin of experience.

\noindent\textbf{The Evolutionary Bridge: From \emph{Place} to \emph{Name}}
In rodents, place coding maps sensory flux to a spatial anchor \cite{moser_place_2008}:
a place code asserts that a complex constellation of cues corresponds to \emph{location $x$}.
In humans, we propose an evolutionary lift in which the same anchoring operation is connected
to social and vocal-gestural systems, enabling a mapping from sensory flux to a \emph{symbolic anchor}.
Formally, the mapping shifts from
$\text{sensation} \;\mapsto\; \text{coordinate}
\quad\text{to}\quad
\text{sensation} \;\mapsto\; \text{symbol}.$
A \emph{word} is a place-cell-like anchor for an \emph{idea}: firing the token ``water'' selects a
basin in semantic space and instructs the listener to navigate to that basin. In this view,
lexical items are not arbitrary labels; they are \emph{indices} that make concept navigation fast,
repeatable, and shareable.

\noindent\textbf{A Phase Transition in the Wild: Helen Keller at the Well}
The well-known ``water'' episode is not only emotionally compelling; it is a rare, well-documented
instance of a \emph{representational phase transition}. The event can be analyzed as a forced coupling
of two input streams with distinct geometry:
\emph{Dual stream coupling.}
At the pump, Keller received two simultaneous signals:
1) \textbf{Referent stream (reality).} Cool, flowing liquid over one hand: a high-dimensional,
    continuous sensory manifold $\Psi$;
2) \textbf{Sign stream (token).} The discrete tactile sequence spelling W-A-T-E-R into the other hand: a low-dimensional, repeatable pattern $\Phi$.
\emph{Pre-transition: disjoint open sets.}
Before the event, both streams were familiar, yet \emph{topologically disjoint}:
the referent was ``a feeling,'' the sign was ``a movement,'' and there was no stable mapping from sign
to referent. In manifold terms, Keller's internal state space contained disconnected regions without
a reliable index relation.
Ms. Sullivan enforced tight temporal coincidence and repetition between the two streams.
Mechanistically, this creates structured co-activation; variationally, the system is driven to reduce
the conflict between repeatedly co-occurring patterns. The key operation is the discovery of an invariant
binding:
$d(\Phi,\Psi)\;\longrightarrow\;0,$
i.e., the metric distance between sign and referent collapses under sustained coincidence.
\emph{Collapse to a symbolic soliton.}
At the transition, the coupled streams condense into a new attractor: the concept \emph{water} becomes
a stable, manipulable handle. The crucial shift is from living in an unindexed stream (continuous
experience) to holding an indexed object (discrete symbol). Symbolic consciousness arises when
reality becomes \emph{addressable} by compact tokens.

\noindent\textbf{Why Modality Does Not Matter: Substrate-Independence of the Operator}
A common intuition is that language requires hearing or vision. Keller's case demonstrates that this
cannot be the essential requirement. The Urysohn account is that the operator is modality-agnostic \cite{finn2017model}:
what matters is the \emph{topological shape} of the signal and its role as an index.
\emph{Carrier vs.\ structure.}
Different sensory channels provide different carriers:
photons to V1, phonons to A1, pressure patterns to S1. Yet a word token, regardless of carrier, has
shared structural properties:
it is discrete, repeatable, and temporally alignable with a referent stream.
Hence, the same condensation mechanism can bind a tactile token to a sensory referent just as it binds
an auditory token to a sensory referent.
\emph{Interpretation.}
Keller's success suggests that intelligence is substrate-independent at the relevant level:
it depends on the geometry and topology of signals (coincidence, invariance, separability),
not on any specific sensory modality.

\noindent\textbf{Language as Shared Topology: Compression, Transmission, Reconstruction}
Why are words useful? Because they enable the social sharing of internal trajectories using low-bandwidth
communication \cite{chomsky_minimalist_2014}.
\emph{Compression.}
A thought is a trajectory through a high-dimensional manifold. Direct transmission of that trajectory is
impossible. Language compresses the trajectory into a compact soliton-like token.
\emph{Transmission.}
The token is broadcast through a narrow channel (speech, sign, touch). This is the social analogue of
``compiling time into space'': extensive internal computation is replaced by an address.
\emph{Reconstruction (inverse Urysohn).}
Upon receiving the token, the listener's system expands it back into a compatible internal trajectory.
Communication succeeds when sender and receiver share enough manifold structure for the token to land in
a homologous basin.

\section{Discussion: The Evolutionary Endpoint of Intelligence}
\label{sec:discussion}

The MTF framework presents a paradoxical view of cognitive evolution: the ultimate goal of biological intelligence is to eliminate the need for active computation. In biological systems, ``thinking''- defined as sequential, topological search through a non-convex state space - is a thermodynamic liability. It consumes massive amounts of ATP and delays reaction times, both of which are penalizing in hostile environments. Therefore, governed by the Principle of Least Action, the evolutionary trajectory of the neocortex is a relentless drive toward perfect metric amortization.

\noindent\textbf{Advice and the Annihilation of ``Time''}
From a computational standpoint, evolution is the billion-year manufacture of \emph{advice} \cite{arora2009computational}. Primitive organisms, lacking pre-positioned structure, must navigate the causal topology of the world by trial-and-error search, paying in time. The expansion of the mammalian neocortex represents a structural phase transition \cite{buzsaki_rhythms_2006}: a massive scaling of the capacity to store advice ($W$) so that the requisite search ($G$) at decision time approaches zero.
We stress that this is not a space-for-time trade in the sense of classical space--time tradeoffs, which redistribute two fungible resources. Advice is a \emph{third} axis, provably not reducible to the other two: $\mathrm{P/poly}$ decides languages that unbounded time and space cannot \cite{arora2009computational}. What consolidation buys is therefore not a faster search but a different resource, and that irreducibility is precisely why a purely metric engine cannot substitute for the indexing-and-condensation cycle. The claim is non-reducibility, not a free lunch: advice must be built, and its cost is paid in the time, space and experience of a prior phase -- evolutionary, developmental, or offline replay \cite{gershman_amortized_2014}. 
Under the MTF framework, a perfectly adapted intelligent agent does not execute internal simulations or tree-search algorithms during inference. Instead, it relies on the \textit{Urysohn Collapse} to pre-flatten the ``Maze'' of reality into a ``Convex Bowl.'' Once the environment is fully amortized into an Equiangular Tight Frame (ETF) \cite{papyan2020prevalence}, decision-making becomes a zero-shot, $\mathcal{O}(1)$ reflex. We observe glimpses of this in human cognition as ``flow states'' or motor automaticity, where the prefrontal cortex quiets, and high-dimensional problem-solving becomes computationally indistinguishable from a simple spinal reflex.

\noindent\textbf{The Thermodynamic Fixed Point}
We can reinterpret this evolutionary endpoint through the lens of Optimal Transport and the Free Energy Principle \cite{friston_free-energy_2010}. The brain continuously seeks to minimize the Wasserstein distance ($\mathcal{W}_2$) between its internal generative metric and the external topology of the universe. 
If an intelligence were to achieve a perfect isomorphic mapping of its environment, prediction error would fall to absolute zero. At this theoretical fixed point, the internal geometry is in perfect thermodynamic equilibrium with reality. The system acts with infinite complexity, yet burns near-zero computational energy doing so, because the chaotic variance of the world has been successfully converted into the silent gravity of a geometric basin.

\noindent\textbf{Implications for Artificial General Intelligence.}
This perspective fundamentally challenges the current trajectory of Artificial General Intelligence (AGI). Modern autoregressive Large Language Models (LLMs) function as massive, disembodied metric tensors, crystallized ``intelligence'' that lacks the conscious engine required for true adaptation. When these frozen policies encounter novel, out-of-distribution (OOD) topology, they cannot autonomously execute the Urysohn Collapse. While current AI paradigms attempt to patch this brittleness by bolting on external search heuristics (e.g., test-time compute or Tree-of-Thoughts), these interventions scale linearly with energy and time, directly violating biological least-action principles \cite{siburg2004principle}. 
Our results suggest that true AGI cannot be achieved merely by scaling dataset size or test-time compute. Instead, artificial systems must be endowed with a closed-loop, cycle-consistent architecture governed by the CCUP \cite{li_content-context_2025}. An AGI must possess the structural capacity to experience the topological frustration of a novel task (the active ``Where''), and subsequently execute offline metric amortization (analogous to biological sleep) to hardcode that task's invariances into its base geometry (the stable ``What''). Furthermore, under the MTF framework, causal embodiment \cite{shapiro_embodied_2019} is not merely advantageous, but a strict mathematical prerequisite for AGI. Without causal situatedness, a system cannot generate the topological friction necessary to smelt new representations. Ultimately, we should not judge the advancement of an AGI by how deeply it can ``think'' through a problem at inference time, but by how effectively it can restructure its own internal geometry so that thinking is no longer required.

\section{Conclusion}

We have argued that intelligence cannot be understood as improved optimization
within a fixed representational landscape.
When environments are semantically complex and topologically shifting,
no single smooth geometry can remain globally contractive.
This fundamental limitation, \emph{geometric incompleteness}, implies that
pure search-based optimization is unstable,
energetically expensive,
and vulnerable to interference under distributional change.
MTF provides a principled resolution.
By separating discrete contextual indexing from continuous geometric shaping,
biological systems avoid forcing incompatible regimes
into a single metric basin.
The hippocampus identifies the active relational manifold;
the neocortex progressively shapes locally contractive geometries within it.
Hierarchical cortical organization implements a sequential decomposition
that overcomes the curse of dimensionality,
converting tangled sensory manifolds into linearly separable representations.
Replay and consolidation further amortize structural resolution offline,
allowing online behavior to operate with reduced computational cost.

Under this framework,
catastrophic interference, remapping, sleep,
and conscious awareness are not disparate phenomena,
but coordinated phases of a unified computational cycle.
Wakefulness contracts experience into stable basins.
Sleep reorganizes and compresses structural regularities.
Consciousness corresponds to moments in which contextual structure
remains unresolved and must be actively selected.
Across timescales, intelligence emerges as the controlled alternation
between exploration of relational topology
and stabilization of geometric structure.
Evolutionary scaling can be reinterpreted through this lens.
From sensorimotor control to allocentric navigation,
episodic memory, social cognition, and language,
each transition reflects expansion in the dimensionality
and connectivity of the relational manifold.
Each expansion required more powerful mechanisms
for separating contextual regimes from local geometric shaping.
Rather than inventing fundamentally new computational principles,
evolution progressively extended relational machinery
into increasingly abstract domains.

The MTF perspective also clarifies the limitations of purely feedforward,
metric-only artificial systems.
Models that operate within a single fixed geometry
may achieve impressive interpolation performance,
yet remain brittle under topological shift.
A robust artificial general intelligence must incorporate
explicit mechanisms for contextual indexing,
hierarchical untangling,
and alternating online–offline structural refinement.
In this sense, the path toward AGI is not deeper search but better control of geometry under changing topology.
Intelligence, therefore, is neither brute-force exploration
nor static representation.
It is the continual recalibration of relational structure
and geometric contraction
across wakefulness, sleep, development, and evolution.
By recognizing the necessity of factorizing topology and metric,
we obtain a unified account of stability–plasticity tradeoffs,
energy efficiency,
hierarchical abstraction,
and conscious awareness
within a single geometric framework.

\bibliographystyle{plain}
\bibliography{references,ref}  

\begin{thebibliography}{100}

\bibitem{abraham2005memory}
Wickliffe~C Abraham and Anthony Robins.
\newblock Memory retention--the synaptic stability versus plasticity dilemma.
\newblock {\em Trends in neurosciences}, 28(2):73--78, 2005.

\bibitem{amari1998natural}
Shun-ichi Amari.
\newblock Natural gradient works efficiently in learning.
\newblock {\em Neural Computation}, 10(2):251--276, 1998.

\bibitem{aronov2017mapping}
Dmitriy Aronov, Randall Nevers, and David~W. Tank.
\newblock Mapping of a non-spatial dimension by the hippocampal–entorhinal circuit.
\newblock {\em Nature}, 543:719--722, 2017.

\bibitem{arora2009computational}
Sanjeev Arora and Boaz Barak.
\newblock {\em Computational Complexity: A Modern Approach}.
\newblock Cambridge University Press, 2009.

\bibitem{ashby1956introduction}
William~Ross Ashby.
\newblock An introduction to cybernetics.
\newblock 1956.

\bibitem{baars1993cognitive}
Bernard~J Baars.
\newblock {\em A cognitive theory of consciousness}.
\newblock Cambridge University Press, 1993.

\bibitem{Badre2007RostroCaudal}
David Badre and Mark D'Esposito.
\newblock The {R}ostro--{C}audal axis of the frontal cortex is hierarchically organized.
\newblock {\em PLoS Biology}, 5(6):e140, 2007.

\bibitem{bastos_canonical_2012}
Andre~M Bastos, W~Martin Usrey, Rick~A Adams, George~R Mangun, Pascal Fries, and Karl~J Friston.
\newblock Canonical microcircuits for predictive coding.
\newblock {\em Neuron}, 76(4):695--711, 2012.

\bibitem{bausch2026distinct}
Marcel Bausch, Johannes Niediek, Thomas~P Reber, Sina Mackay, Jan Bostr{\"o}m, Christian~E Elger, and Florian Mormann.
\newblock Distinct neuronal populations in the human brain combine content and context.
\newblock {\em Nature}, pages 1--11, 2026.

\bibitem{behrens_what_2018}
Timothy~EJ Behrens, Timothy~H Muller, James~CR Whittington, S~Mark, AB~Baram, Kimberly~L Stachenfeld, and {others}.
\newblock What is a cognitive map? {Organizing} knowledge for flexible behavior.
\newblock {\em Neuron}, 100(2):490--509, 2018.

\bibitem{bellman_dynamic_1966}
Richard Bellman.
\newblock Dynamic programming.
\newblock {\em science}, 153(3731):34--37, 1966.

\bibitem{bengio_representation_2013}
Yoshua Bengio, Aaron Courville, and Pascal Vincent.
\newblock Representation learning: {A} review and new perspectives.
\newblock In {\em {IEEE} transactions on pattern analysis and machine intelligence}, volume~35, pages 1798--1828. IEEE, 2013.

\bibitem{bennett_brief_2023}
Max Bennett.
\newblock {\em A brief history of intelligence: evolution, {AI}, and the five breakthroughs that made our brains}.
\newblock HarperCollins, 2023.

\bibitem{botvinick2004conflict}
Matthew~M. Botvinick, Jonathan~D. Cohen, and Cameron~S. Carter.
\newblock Conflict monitoring and cognitive control.
\newblock {\em Psychological Review}, 108(3):624--652, 2001.

\bibitem{brown2020language}
Tom Brown, Benjamin Mann, Nick Ryder, Melanie Subbiah, Jared~D Kaplan, Prafulla Dhariwal, Arvind Neelakantan, Pranav Shyam, Girish Sastry, Amanda Askell, et~al.
\newblock Language models are few-shot learners.
\newblock {\em Advances in neural information processing systems}, 33:1877--1901, 2020.

\bibitem{buzsaki2015hippocampal}
Gy{\"o}rgy Buzs{\'a}ki.
\newblock Hippocampal sharp wave-ripple: A cognitive biomarker for episodic memory and planning.
\newblock {\em Hippocampus}, 25(10):1073--1188, 2015.

\bibitem{buzsaki_hippocampo-neocortical_1996}
György Buzsáki.
\newblock The hippocampo-neocortical dialogue.
\newblock {\em Cerebral cortex}, 6(2):81--92, 1996.

\bibitem{buzsaki_rhythms_2006}
György Buzsáki.
\newblock {\em Rhythms of the {Brain}}.
\newblock Oxford university press, 2006.

\bibitem{buzsaki_memory_2013}
György Buzsáki and Edvard~I Moser.
\newblock Memory, navigation and theta rhythm in the hippocampal-entorhinal system.
\newblock {\em Nature neuroscience}, 16(2):130--138, 2013.

\bibitem{chi2019nonconvex}
Yuejie Chi, Yue~M Lu, and Yuxin Chen.
\newblock Nonconvex optimization meets low-rank matrix factorization: An overview.
\newblock {\em IEEE Transactions on Signal Processing}, 67(20):5239--5269, 2019.

\bibitem{chomsky_three_1956}
Noam Chomsky.
\newblock Three models for the description of language.
\newblock {\em IRE Transactions on information theory}, 2(3):113--124, 1956.

\bibitem{chomsky_minimalist_2014}
Noam Chomsky.
\newblock {\em The minimalist program}.
\newblock MIT press, 2014.

\bibitem{colgin_understanding_2008}
Laura~Lee Colgin, Edvard~I Moser, and May-Britt Moser.
\newblock Understanding memory through hippocampal remapping.
\newblock {\em Trends in neurosciences}, 31(9):469--477, 2008.

\bibitem{constantinescu_organizing_2016}
Alexandra~O. Constantinescu, Jill~X. O'Reilly, and Timothy E.~J. Behrens.
\newblock Organizing {Conceptual} {Knowledge} in {Humans} with a {Gridlike} {Code}.
\newblock {\em Science}, 352(6292):1464--1468, 2016.

\bibitem{constantinescu2016organizing}
Andreea~O. Constantinescu, Jill~X. O'Reilly, and Timothy E.~J. Behrens.
\newblock Organizing conceptual knowledge in humans with a gridlike code.
\newblock {\em Science}, 352(6292):1464--1468, 2016.

\bibitem{danilova2022recent}
Marina Danilova, Pavel Dvurechensky, Alexander Gasnikov, Eduard Gorbunov, Sergey Guminov, Dmitry Kamzolov, and Innokentiy Shibaev.
\newblock Recent theoretical advances in non-convex optimization.
\newblock In {\em High-Dimensional Optimization and Probability: With a View Towards Data Science}, pages 79--163. Springer, 2022.

\bibitem{dauphin2014identifying}
Yann~N. Dauphin, Razvan Pascanu, Caglar Gulcehre, Kyunghyun Cho, Surya Ganguli, and Yoshua Bengio.
\newblock Identifying and attacking the saddle point problem in high-dimensional non-convex optimization.
\newblock {\em Advances in Neural Information Processing Systems}, 27, 2014.

\bibitem{daw2005uncertainty}
Nathaniel~D. Daw, Yael Niv, and Peter Dayan.
\newblock Uncertainty-based competition between prefrontal and dorsolateral striatal systems for behavioral control.
\newblock {\em Nature Neuroscience}, 8(12):1704--1711, 2005.

\bibitem{dayan_theoretical_2001}
Peter Dayan and L.~F. Abbott.
\newblock {\em Theoretical {Neuroscience}: {Computational} and {Mathematical} {Modeling} of {Neural} {Systems}}.
\newblock MIT Press, 2001.

\bibitem{dehaene2011experimental}
Stanislas Dehaene and Jean-Pierre Changeux.
\newblock Experimental and theoretical approaches to conscious processing.
\newblock {\em Neuron}, 70(2):200--227, 2011.

\bibitem{dicarlo_untangling_2007}
James~J DiCarlo and David~D Cox.
\newblock Untangling invariant object recognition.
\newblock {\em Trends in cognitive sciences}, 11(8):333--341, 2007.

\bibitem{dicarlo_how_2012}
James~J DiCarlo, Davide Zoccolan, and Nicole~C Rust.
\newblock How does the brain solve visual object recognition?
\newblock {\em Neuron}, 73(3):415--434, 2012.

\bibitem{diekelmann_memory_2010}
Susanne Diekelmann and Jan Born.
\newblock The memory function of sleep.
\newblock {\em Nature Reviews Neuroscience}, 11(2):114--126, 2010.

\bibitem{diekelmann2010memory}
Susanne Diekelmann and Jan Born.
\newblock The memory function of sleep.
\newblock {\em Nature Reviews Neuroscience}, 11:114--126, 2010.

\bibitem{Douglas2004NeuronalCircuits}
Rodney~J. Douglas and Kevan A.~C. Martin.
\newblock Neuronal circuits of the neocortex.
\newblock {\em Annual Review of Neuroscience}, 27:419--451, 2004.

\bibitem{dunbar_grooming_1996}
Robin Ian~MacDonald Dunbar.
\newblock {\em Grooming, gossip, and the evolution of language}.
\newblock Harvard University Press, 1996.

\bibitem{dunbar1992neocortex}
Robin~IM Dunbar.
\newblock Neocortex size as a constraint on group size in primates.
\newblock {\em Journal of human evolution}, 22(6):469--493, 1992.

\bibitem{edelman_universe_2008}
Gerald~M Edelman and Giulio Tononi.
\newblock {\em A universe of consciousness: {How} matter becomes imagination}.
\newblock Basic books, 2008.

\bibitem{edelsbrunner_computational_2010}
Herbert Edelsbrunner and John Harer.
\newblock {\em Computational {Topology}: {An} {Introduction}}.
\newblock American Mathematical Society, 2010.

\bibitem{epstein2017cognitive}
Russell~A Epstein, Eva~Zita Patai, Joshua~B Julian, and Hugo~J Spiers.
\newblock The cognitive map in humans: spatial navigation and beyond.
\newblock {\em Nature neuroscience}, 20(11):1504--1513, 2017.

\bibitem{Felleman1991Distributed}
Daniel~J. Felleman and David~C. Van~Essen.
\newblock Distributed hierarchical processing in the primate cerebral cortex.
\newblock {\em Cerebral Cortex}, 1(1):1--47, 1991.

\bibitem{finn2017model}
Chelsea Finn, Pieter Abbeel, and Sergey Levine.
\newblock Model-agnostic meta-learning for fast adaptation of deep networks.
\newblock In {\em International conference on machine learning}, pages 1126--1135. PMLR, 2017.

\bibitem{frank2006sleep}
Marcos~G. Frank and Joel~H. Benington.
\newblock The role of sleep in memory consolidation and brain plasticity: dream or reality?
\newblock {\em The Neuroscientist}, 12(6):477--488, 2006.

\bibitem{french1999catastrophic}
Robert~M French.
\newblock Catastrophic forgetting in connectionist networks.
\newblock {\em Trends in cognitive sciences}, 3(4):128--135, 1999.

\bibitem{freud2016happening}
Erez Freud, David~C Plaut, and Marlene Behrmann.
\newblock ‘what’is happening in the dorsal visual pathway.
\newblock {\em Trends in cognitive sciences}, 20(10):773--784, 2016.

\bibitem{friston_free-energy_2010}
Karl Friston.
\newblock The free-energy principle: a unified brain theory?
\newblock {\em Nature Reviews Neuroscience}, 11(2):127--138, 2010.

\bibitem{friston_active_2017}
Karl Friston, Francesco Rigoli, David Ognibene, Christoph Mathys, Thomas FitzGerald, Giovanni Pezzulo, and {...}
\newblock Active inference: a process theory.
\newblock {\em Neural Computation}, 29(1):1--49, 2017.

\bibitem{friston_action_2010}
Karl~J Friston, Jean Daunizeau, James Kilner, and Stefan~J Kiebel.
\newblock Action and behavior: a free-energy formulation.
\newblock {\em Biological cybernetics}, 102:227--260, 2010.

\bibitem{frith_social_2007}
Chris~D. Frith and Uta Frith.
\newblock The social brain?
\newblock {\em Philosophical Transactions of the Royal Society B: Biological Sciences}, 362(1480):671--678, 2007.

\bibitem{fyhn2007hippocampal}
Marianne Fyhn, Torkel Hafting, Alessandro Treves, May-Britt Moser, and Edvard~I Moser.
\newblock Hippocampal remapping and grid realignment in entorhinal cortex.
\newblock {\em Nature}, 446(7132):190--194, 2007.

\bibitem{ge2015escaping}
Rong Ge, Furong Huang, Chi Jin, and Yang Yuan.
\newblock Escaping from saddle points---online stochastic gradient for tensor decomposition.
\newblock In {\em Conference on Learning Theory (COLT)}, pages 797--842, 2015.

\bibitem{geman_neural_1992}
Stuart Geman, Elie Bienenstock, and René Doursat.
\newblock Neural networks and the bias/variance dilemma.
\newblock {\em Neural computation}, 4(1):1--58, 1992.

\bibitem{gershman_amortized_2014}
Samuel Gershman and Noah Goodman.
\newblock Amortized inference in probabilistic reasoning.
\newblock In {\em Proceedings of the 36th {Annual} {Meeting} of the {Cognitive} {Science} {Society}}, 2014.

\bibitem{goodale_separate_1992}
Melvyn~A. Goodale and A.~David Milner.
\newblock Separate visual pathways for perception and action.
\newblock {\em Trends in Neurosciences}, 15(1):20--25, 1992.

\bibitem{gould1982exaptation}
Stephen~Jay Gould and Elisabeth~S Vrba.
\newblock Exaptation—a missing term in the science of form.
\newblock {\em Paleobiology}, 8(1):4--15, 1982.

\bibitem{godel_formally_1962}
Kurt Gödel.
\newblock {\em On {Formally} {Undecidable} {Propositions} of {\textbackslash}em {Principia} {Mathematica} and {Related} {Systems}}.
\newblock Dover Publications, New York, 1962.

\bibitem{hatcher_algebraic_2002}
Allen Hatcher.
\newblock {\em Algebraic topology}.
\newblock Cambridge University Press, 2002.

\bibitem{hebb2005organization}
Donald~Olding Hebb.
\newblock {\em The organization of behavior: A neuropsychological theory}.
\newblock Psychology press, 2005.

\bibitem{hickok2004dorsal}
Gregory Hickok and David Poeppel.
\newblock Dorsal and ventral streams: a framework for understanding aspects of the functional anatomy of language.
\newblock {\em Cognition}, 92(1-2):67--99, 2004.

\bibitem{hinton__1995}
Geoffrey~E Hinton, Peter Dayan, Brendan~J Frey, and Radford~M Neal.
\newblock The" wake-sleep" algorithm for unsupervised neural networks.
\newblock {\em Science}, 268(5214):1158--1161, 1995.

\bibitem{hirsch1976differential}
Morris~W. Hirsch.
\newblock {\em Differential Topology}.
\newblock Springer, 1976.

\bibitem{hobson1977brain}
J~Allan Hobson and Robert~W McCarley.
\newblock The brain as a dream state generator: an activation-synthesis hypothesis of the dream process.
\newblock {\em The American journal of psychiatry}, 134(12):1335--1348, 1977.

\bibitem{hobson1975sleep}
J~Allan Hobson, Robert~W McCarley, and Peter~W Wyzinski.
\newblock Sleep cycle oscillation: reciprocal discharge by two brainstem neuronal groups.
\newblock {\em Science}, 189(4196):55--58, 1975.

\bibitem{jin2017escape}
Chi Jin, Rong Ge, Praneeth Netrapalli, Sham~M. Kakade, and Michael~I. Jordan.
\newblock How to escape saddle points efficiently.
\newblock In {\em International Conference on Machine Learning (ICML)}, pages 1724--1732, 2017.

\bibitem{Kaas2000Subdivisions}
Jon~H. Kaas and Troy~A. Hackett.
\newblock Subdivisions of auditory cortex and processing streams in primates.
\newblock {\em Proceedings of the National Academy of Sciences}, 97(22):11793--11799, 2000.

\bibitem{keller_story_1903}
Helen Keller.
\newblock {\em The {Story} of {My} {Life}}.
\newblock Doubleday, Page \& Company, 1903.

\bibitem{kirkpatrick_overcoming_2017}
James Kirkpatrick, Razvan Pascanu, Neil Rabinowitz, Joel Veness, Guillaume Desjardins, Andrei~A Rusu, Kieran Milan, John Quan, Tiago Ramalho, Agnieszka Grabska-Barwinska, and {others}.
\newblock Overcoming catastrophic forgetting in neural networks.
\newblock {\em Proceedings of the national academy of sciences}, 114(13):3521--3526, 2017.

\bibitem{kubie2020hippocampal}
John~L Kubie, Eliott~RJ Levy, and Andr{\'e}~A Fenton.
\newblock Is hippocampal remapping the physiological basis for context?
\newblock {\em Hippocampus}, 30(8):851--864, 2020.

\bibitem{kulis2013metric}
Brian Kulis.
\newblock Metric learning: A survey.
\newblock {\em Foundations and Trends{\textregistered} in Machine Learning}, 5(4):287--364, 2013.

\bibitem{kumaran2007match}
Dharshan Kumaran and Eleanor~A. Maguire.
\newblock Match–mismatch processes underlie human hippocampal responses to associative novelty.
\newblock {\em Journal of Neuroscience}, 27(32):8517--8524, 2007.

\bibitem{lee_hierarchical_2003}
Tai~Sing Lee and David Mumford.
\newblock Hierarchical {Bayesian} inference in the visual cortex.
\newblock {\em Journal of the Optical Society of America A}, 20(7):1434--1448, 2003.

\bibitem{li_beyond_2025}
Xin Li.
\newblock Beyond {Turing}: {Memory}-{Amortized} {Inference} as a {Foundation} for {Cognitive} {Computation}.
\newblock {\em ArXiv}, 2025.

\bibitem{li_content-context_2025}
Xin Li.
\newblock On {Content}-{Context} {Uncertainty} {Principle}.
\newblock {\em Neural Information Processing Symposium}, 2025.

\bibitem{li2026beyond}
Xin Li.
\newblock Beyond optimization: Intelligence as metric-topology factorization under geometric incompleteness.
\newblock {\em arXiv preprint arXiv:2602.07974}, 2026.

\bibitem{lycan1996consciousness}
William~G Lycan.
\newblock {\em Consciousness and experience}.
\newblock Mit Press, 1996.

\bibitem{mashour_conscious_2020}
George~A Mashour, Pieter~R Roelfsema, Jean-Pierre Changeux, and Stanislas Dehaene.
\newblock Conscious processing and the global neuronal workspace hypothesis.
\newblock {\em Neuron}, 105(5):776--798, 2020.

\bibitem{mcclelland_why_1995}
James~L. McClelland, Bruce~L. McNaughton, and Randall~C. O'Reilly.
\newblock Why there are complementary learning systems in the hippocampus and neocortex: {Insights} from the successes and failures of connectionist models of learning and memory.
\newblock {\em Psychological Review}, 102(3):419--457, 1995.

\bibitem{mccloskey_catastrophic_1989}
Michael McCloskey and Neal~J Cohen.
\newblock Catastrophic interference in connectionist networks: {The} sequential learning problem.
\newblock In {\em Psychology of learning and motivation}, volume~24, pages 109--165. Elsevier, 1989.

\bibitem{mcnaughton_dead_1991}
BL~McNaughton, LL~Chen, and EJ~Markus.
\newblock “{Dead} reckoning,” landmark learning, and the sense of direction: a neurophysiological and computational hypothesis.
\newblock {\em Journal of cognitive neuroscience}, 3(2):190--202, 1991.

\bibitem{mcnaughton_path_2006}
Bruce~L McNaughton, Francesco~P Battaglia, Ole Jensen, Edvard~I Moser, and May-Britt Moser.
\newblock Path integration and the neural basis of the'cognitive map'.
\newblock {\em Nature Reviews Neuroscience}, 7(8):663--678, 2006.

\bibitem{medin1981linear}
Douglas~L Medin and Paula~J Schwanenflugel.
\newblock Linear separability in classification learning.
\newblock {\em Journal of Experimental Psychology: Human Learning and Memory}, 7(5):355, 1981.

\bibitem{merel2019hierarchical}
Josh Merel, Matthew Botvinick, and Greg Wayne.
\newblock Hierarchical motor control in mammals and machines.
\newblock {\em Nature communications}, 10(1):5489, 2019.

\bibitem{milnor1963morse}
John~Willard Milnor.
\newblock {\em Morse theory}.
\newblock Number~51. Princeton university press, 1963.

\bibitem{minsky_steps_1961}
Marvin Minsky.
\newblock Steps toward artificial intelligence.
\newblock {\em Proceedings of the IRE}, 49(1):8--30, 1961.

\bibitem{moser_place_2008}
Edvard~I Moser, Emilio Kropff, and May-Britt Moser.
\newblock Place cells, grid cells, and the brain's spatial representation system.
\newblock {\em Annu. Rev. Neurosci.}, 31(1):69--89, 2008.

\bibitem{mountcastle_columnar_1997}
Vernon~B Mountcastle.
\newblock The columnar organization of the neocortex.
\newblock {\em Brain: a journal of neurology}, 120(4):701--722, 1997.

\bibitem{mountcastle1998perceptual}
Vernon~B Mountcastle.
\newblock {\em Perceptual neuroscience: The cerebral cortex}.
\newblock Harvard University Press, 1998.

\bibitem{muller_hippocampus_1996}
Robert~U Muller, Matt Stead, and Janos Pach.
\newblock The hippocampus as a cognitive graph.
\newblock {\em The Journal of general physiology}, 107(6):663--694, 1996.

\bibitem{nakahara_geometry_2018}
Mikio Nakahara.
\newblock {\em Geometry, topology and physics}.
\newblock CRC press, 2018.

\bibitem{okeefe_hippocampus_1971}
John O'Keefe and Jonathan Dostrovsky.
\newblock The {Hippocampus} as a {Spatial} {Map}. {Preliminary} {Evidence} from {Unit} {Activity} in the {Freely}-{Moving} {Rat}.
\newblock {\em Brain Research}, 34(1):171--175, 1971.

\bibitem{okeefe_hippocampus_1978}
John O'keefe and Lynn Nadel.
\newblock {\em The hippocampus as a cognitive map}.
\newblock Oxford university press, 1978.

\bibitem{papyan2020prevalence}
Vardan Papyan, XY~Han, and David~L Donoho.
\newblock Prevalence of neural collapse during the terminal phase of deep learning training.
\newblock {\em Proceedings of the National Academy of Sciences}, 117(40):24652--24663, 2020.

\bibitem{revonsuo2000reinterpretation}
Antti Revonsuo.
\newblock The reinterpretation of dreams: An evolutionary hypothesis of the function of dreaming.
\newblock {\em Behavioral and brain sciences}, 23(6):877--901, 2000.

\bibitem{RiesenhuberPoggio1999Hierarchical}
Maximilian Riesenhuber and Tomaso Poggio.
\newblock Hierarchical models of object recognition in cortex.
\newblock {\em Nature Neuroscience}, 2(11):1019--1025, 1999.

\bibitem{russell1995modern}
Stuart Russell, Peter Norvig, and Artificial Intelligence.
\newblock A modern approach.
\newblock {\em Artificial Intelligence. Prentice-Hall, Egnlewood Cliffs}, 25(27):79--80, 1995.

\bibitem{russo2018motor}
Abigail~A Russo, Sean~R Bittner, Sean~M Perkins, Jeffrey~S Seely, Brian~M London, Antonio~H Lara, Andrew Miri, Najja~J Marshall, Adam Kohn, Thomas~M Jessell, et~al.
\newblock Motor cortex embeds muscle-like commands in an untangled population response.
\newblock {\em Neuron}, 97(4):953--966, 2018.

\bibitem{shapiro_embodied_2019}
Lawrence Shapiro.
\newblock {\em Embodied cognition}.
\newblock Routledge, 2019.

\bibitem{shiffrin1977controlled}
Richard~M. Shiffrin and Walter Schneider.
\newblock Controlled and automatic human information processing.
\newblock {\em Psychological Review}, 84(2):127--190, 1977.

\bibitem{siburg2004principle}
Karl~Friedrich Siburg.
\newblock {\em The principle of least action in geometry and dynamics}.
\newblock Number 1844. Springer Science \& Business Media, 2004.

\bibitem{simon2012architecture}
Herbert~A Simon.
\newblock The architecture of complexity.
\newblock In {\em The Roots of Logistics}, pages 335--361. Springer, 2012.

\bibitem{simoncelli2001natural}
Eero~P Simoncelli and Bruno~A Olshausen.
\newblock Natural image statistics and neural representation.
\newblock {\em Annual review of neuroscience}, 24(1):1193--1216, 2001.

\bibitem{sohn_network_2021}
Hansem Sohn, Nicolas Meirhaeghe, Rishi Rajalingham, and Mehrdad Jazayeri.
\newblock A network perspective on sensorimotor learning.
\newblock {\em Trends in Neurosciences}, 44(3):170--181, 2021.

\bibitem{squire_memory_2015}
Larry~R Squire, Lisa Genzel, John~T Wixted, and Richard~G Morris.
\newblock Memory consolidation.
\newblock {\em Cold Spring Harbor perspectives in biology}, 7(8):a021766, 2015.

\bibitem{squire_structure_1993}
Larry~R Squire, Barbara Knowlton, and Gail Musen.
\newblock The structure and organization of memory.
\newblock {\em Annual review of psychology}, 44(1):453--495, 1993.

\bibitem{stickgold_sleep-dependent_2005}
Robert Stickgold.
\newblock Sleep-dependent memory consolidation.
\newblock {\em Nature}, 437(7063):1272--1278, 2005.

\bibitem{sun2024easy}
Zhiqing Sun, Longhui Yu, Yikang Shen, Weiyang Liu, Yiming Yang, Sean Welleck, and Chuang Gan.
\newblock Easy-to-hard generalization: Scalable alignment beyond human supervision.
\newblock {\em Advances in Neural Information Processing Systems}, 37:51118--51168, 2024.

\bibitem{tarsky1956concept}
A~Tarsky.
\newblock The concept of truth in formalized languages.
\newblock {\em Logic, semantics, metamathematics}, pages 152--278, 1956.

\bibitem{teyler1986hippocampal}
Timothy~J Teyler and Pascal DiScenna.
\newblock The hippocampal memory indexing theory.
\newblock {\em Behavioral neuroscience}, 100(2):147, 1986.

\bibitem{tolman_cognitive_1948}
Edward~C. Tolman.
\newblock Cognitive {Maps} in {Rats} and {Men}.
\newblock {\em Psychological Review}, 55(4):189--208, 1948.

\bibitem{tononi2004information}
Giulio Tononi.
\newblock An information integration theory of consciousness.
\newblock {\em BMC neuroscience}, 5(1):42, 2004.

\bibitem{tononi_integrated_2016}
Giulio Tononi, Melanie Boly, Marcello Massimini, and Christof Koch.
\newblock Integrated information theory: from consciousness to its physical substrate.
\newblock {\em Nature reviews neuroscience}, 17(7):450--461, 2016.

\bibitem{tononi2003sleep}
Giulio Tononi and Chiara Cirelli.
\newblock Sleep and synaptic homeostasis: a hypothesis.
\newblock {\em Brain Research Bulletin}, 62(2):143--150, 2003.

\bibitem{tononi2014sleep}
Giulio Tononi and Chiara Cirelli.
\newblock Sleep and the price of plasticity.
\newblock {\em Neuron}, 81(1):12--34, 2014.

\bibitem{trappenberg2009fundamentals}
Thomas Trappenberg.
\newblock {\em Fundamentals of computational neuroscience}.
\newblock OUP Oxford, 2009.

\bibitem{tse_schemas_2007}
Dorothy Tse, Rosamund~F Langston, Masaki Kakeyama, Ingrid Bethus, Patrick~A Spooner, Emma~R Wood, Menno~P Witter, and Richard~GM Morris.
\newblock Schemas and memory consolidation.
\newblock {\em Science}, 316(5821):76--82, 2007.

\bibitem{tulving_episodic_2002}
Endel Tulving.
\newblock Episodic memory: {From} mind to brain.
\newblock {\em Annual review of psychology}, 53(1):1--25, 2002.

\bibitem{ungerleider_whatand_1994}
Leslie~G Ungerleider and James~V Haxby.
\newblock ‘{What}’and ‘where’in the human brain.
\newblock {\em Current opinion in neurobiology}, 4(2):157--165, 1994.

\bibitem{urysohn_zum_1925}
Paul Urysohn.
\newblock Zum {Metrizationsproblem}.
\newblock {\em Mathematische Annalen}, 94(1):309--315, 1925.

\bibitem{wang_comprehensive_2024}
Liyuan Wang, Xingxing Zhang, Hang Su, and Jun Zhu.
\newblock A comprehensive survey of continual learning: {Theory}, method and application.
\newblock {\em IEEE transactions on pattern analysis and machine intelligence}, 46(8):5362--5383, 2024.

\bibitem{willard_general_2012}
Stephen Willard.
\newblock {\em General topology}.
\newblock Courier Corporation, 2012.

\bibitem{wilson1994reactivation}
Matthew~A. Wilson and Bruce~L. McNaughton.
\newblock Reactivation of hippocampal ensemble memories during sleep.
\newblock {\em Science}, 265(5172):676--679, 1994.

\bibitem{wilson_reactivation_1994}
Matthew~A. Wilson and Bruce~L. McNaughton.
\newblock Reactivation of hippocampal ensemble memories during sleep.
\newblock {\em Science}, 265(5172):676--679, 1994.

\bibitem{yamins_using_2016}
Daniel L.~K. Yamins and James~J. DiCarlo.
\newblock Using goal-driven deep learning models to understand sensory cortex.
\newblock {\em Nature Neuroscience}, 19(3):356--365, 2016.

\bibitem{YaminsDiCarlo2016Using}
Daniel L.~K. Yamins and James~J. DiCarlo.
\newblock Using goal-driven deep learning models to understand sensory cortex.
\newblock {\em Nature Neuroscience}, 19(3):356--365, 2016.

\bibitem{yang2024generalized}
Jingkang Yang, Kaiyang Zhou, Yixuan Li, and Ziwei Liu.
\newblock Generalized out-of-distribution detection: A survey.
\newblock {\em International Journal of Computer Vision}, 132(12):5635--5662, 2024.

\bibitem{yoo2018economic}
Seng Bum~Michael Yoo and Benjamin~Yost Hayden.
\newblock Economic choice as an untangling of options into actions.
\newblock {\em Neuron}, 99(3):434--447, 2018.

\bibitem{zhu2021geometric}
Zhihui Zhu, Tianyu Ding, Jinxin Zhou, Xiao Li, Chong You, Jeremias Sulam, and Qing Qu.
\newblock A geometric analysis of neural collapse with unconstrained features.
\newblock {\em Advances in Neural Information Processing Systems}, 34:29820--29834, 2021.

\end{thebibliography}


\end{document}